\documentclass[12pt]{article}

\usepackage[margin=2.5cm]{geometry}

\usepackage{amsmath,amsthm,amsfonts,amscd,amssymb,bbm,mathrsfs,enumerate,url}
\usepackage{enumerate}
\usepackage{graphicx,tikz}
\usepackage[pdfborder={0 0 0}]{hyperref}
\usetikzlibrary{arrows}

\def\RR{{\mathbb R}}
\def\CC{{\mathbb C}}

\def\SS{{\mathbb S}}

\def\A{{\mathcal A}}

\def\D{{\mathcal D}}

\def\H{{\mathcal H}}

\def\P{{\mathcal P}}
\def\R{{\mathcal R}}

\def\a{\alpha}
\def\b{\beta}

\def\g{\gamma}

\def\k{\kappa}

\def\l{\lambda}

\def\L{{\mathrm L}}

\def\R{{\mathrm R}}

\def\1{{\mathbbm 1}}

\def\u1net{{\A^{(0)}}}

\def\diff{{\rm Diff}}

\def\diffs1{\diff(S^1)}
\def\mob{{\rm M\ddot{o}b}}
\def\mob2{{\rm M\ddot{o}b}^{(2)}}

\def\psl2r{{\rm PSL}(2,\RR)}
\def\sl2r{{\rm SL}(2,\RR)}
\def\su11{{\rm SU}(1,1)}
\def\2dmob{{\overline{\psl2r}\times\overline{\psl2r}}}
\def\<{\langle}
\def\>{\rangle}

\def\im{\mathrm{Im}\,}
\def\res{\mathrm{Res}\,}

\def\poincare{{\P^\uparrow_+}}

\def\dom{{\mathrm{Dom}}}
\def\fct{\widetilde{\phi}}

\newtheorem{theorem}{Theorem}[section]

\theoremstyle{remark}
\newtheorem{remark}[theorem]{Remark}

\title{Wedge-local observables in the deformed sine-Gordon model}
\date{}
\author{
{\bf Daniela Cadamuro} \\
e-mail: {\tt daniela.cadamuro@mathematik.uni-goettingen.de}\\
Mathematisches Institut, Universit\"at G\"ottingen \\
Bunsenstrasse 3-5, D-37073 G\"ottingen, Germany.  \\
{\bf Yoh Tanimoto}\\
e-mail: {\tt hoyt@mat.uniroma2.it}\\
Dipartimento di Matematica, Universit\`a di Roma ``Tor Vergata''\\
Via della Ricerca Scientifica 1, 00133 Roma, Italy\\
}
\begin{document}
\maketitle
\begin{abstract}
In the bootstrap approach to integrable quantum field theories in the $(1+1)$-dimensional Minkowski space,
one conjectures the two-particle S-matrix and tries to study local observables.
The massless sine-Gordon model is conjectured to be equivalent to the Thirring model, and
its breather-breather S-matrix components (where the first breather corresponds to the scalar field of the sine-Gordon model)
are closed under fusion.
Yet, the residues of the poles in this breather-breather S-matrix have wrong signs
and cannot be considered as a separate model.

We find CDD factors which adjust the signs, so that the breather-breather S-matrix alone
satisfies reasonable assumptions.
Then we propose candidates for observables in wedge-shaped regions
and prove their commutativity in the weak sense.
\end{abstract}

\begin{center}
{\it
Dedicated to Karl-Henning Rehren on the occasion of his 60th birthday
}
\end{center}

\section{Introduction}\label{introduction}

Recently there have been progresses in the construction of $(1+1)$-dimensional
quantum field theories with factorizing S-matrices in the operator algebraic approach
\cite{Lechner03, Lechner08, DT11, Tanimoto12-2, BT13, LST13, LS14, Tanimoto14-1, Alazzawi14, BT15, AL16}.
The basic idea is the following \cite{Schroer97}:
while pointlike local observables are hard to construct,
observables localized in an infinitely extended wedge-shaped region might be tractable
and have simple expressions.
It has been first implemented for a scalar analytic factorizing S-matrix \cite{Lechner08, Alazzawi14, AL16}
and strictly local observables have been shown to exist
using a quite indirect proof that relies on properties of the underlying modular operators
(for double cones above the minimal size).
In this construction, the input is the particle spectrum of the theory, together with the S-matrix with certain properties.
Construction of observables in wedges has been extended to theories with several particle species
by Lechner and Sch\"utzenhofer \cite{LS14}, including the $O(N)$-invariant nonlinear $\sigma$-models.

Recently, in \cite{CT15-1, CT16-diag, Tanimoto16-1}, we further generalized this construction to scalar models with S-matrices
which have poles in the physical strip. The poles in the S-matrix are believed to correspond to the presence of bound states (e.g.\! the Bullough-Dodd model). We also extended this construction to models with several particle species, where the S-matrix is ``diagonal'' in a certain sense. They include, e.g.\! the $Z(N)$-Ising model and the $A_N$-affine Toda field theories.

In this work, we extend this last mentioned class of S-matrices to those which are modifications
of the S-matrix of the massless sine-Gordon model by a CDD factor, and moreover restricted to a certain range
of values for the coupling constant. This is again of diagonal type. 

The massless sine-Gordon model has been conjectured to be equivalent to the Thirring model in a certain sense (Coleman's equivalence). In \cite{BFM09} Benfatto, Falco and Matropietro proved the equivalence between the massless sine-Gordon model with finite volume interaction and the Thirring model with a finite volume mass term. 
The Thirring model has been also constructed by the functional integral methods \cite{BFM07}.
On the other hand, the massless sine-Gordon model has been expected to be integrable and
its S-matrix has been conjectured \cite{ZZ79}.
Yet, in the rigorous constructions, the factorization of the S-matrix has not been proved
(c.f.\! \cite{BR16}, where the perturbative S-matrix with IR cutoff is shown to converge,
yet its factorization has not been proved).

The conjectured S-matrix of the massless sine-Gordon model has been studied in the \textit{form factor programme} \cite{BFKZ99, BK02}.
Certain matrix components of the pointlike local fields (``form factors'') have been computed, yet
the existence of the Wightman field are currently out of reach, because the expansion of
the $n$-point functions in terms of form factors is not under control.
Here, we are not dealing with the massless sine-Gordon model itself, but with a new model with the same fusion structure,
that has not been considered before. It arises as a deformation of the ``breather-breather'' S-matrix of the massless sine-Gordon model
by the multiplication of a CDD factor.
The coupling constant is restricted here to a certain range of values, where there are only two species of particles involved (two breathers).

Our goal is to attain a realization of this model associated with this new S-matrix in the operator-algebraic framework, i.e.\! the Haag-Kastler axioms.
In this framework, we construct candidates for wedge-local observables by extending the construction carried out in \cite{CT16-diag}.
The question of strong commutativity remains open also in this model.

With the presence of poles in the S-matrix,
the construction of wedge-local observables must be studied in a case-by-case approach,
in contrast to the homogeneous construction for the analytic S-matrices \cite{LS14}.
This is due to the idea that simple poles in the S-matrix correspond to the bound states in the model,
therefore, the wedge-local observables must reflect such fusion processes.
We do this by introducing the operators which we call the bound state operators.
Furthermore, higher order poles bring further complications and we need the existence of
what we call elementary particles. Our proof of wedge-locality is based on a number of properties
of the two-particle scattering function, and there is actually a infinite family of examples satisfying
them, therefore, we have correspondingly an infinite family of candidates for quantum field theories.

The paper is organized as follows. In Sec.\! \ref{model}, we will introduce the model and fix the input scattering data, including the properties of the S-matrix. In Sec.\! \ref{Fock}, we exhibit our general notation for multi-particle Fock space, partially following Lechner-Sch\"utzenhofer \cite{LS14}. 
In Sec.\! \ref{chi} we introduce the bound state operators $\chi(f)$, $\chi'(g)$, we analyse their domains and symmetry properties as quadratic forms.
In Sec.\! \ref{sec:comm} we construct the fields $\fct(f)$ and $\fct'(g)$ and show the weak
wedge-commutativity between the components for ``elementary particles''. In Sec.\! \ref{concluding} we conclude our paper
with some remarks.

\section{The deformed integrable sine-Gordon model}\label{model}
\subsection{The factorizing S-matrix}\label{S-matrix}
In the conjectured integrable sine-Gordon model, the particle spectrum consists of
a family of finitely many particles called \emph{breathers} $\{b_\ell\}$ \cite{BK02}.
It is also conjectured that, the sine-Gordon model is equivalent to the Thirring model, 
where the breathers are the bound states of \emph{soliton} and the \emph{anti-soliton} (the anti-particle of the soliton).

In the sine-Gordon model, the number of breathers depends on the coupling constant $0 < \nu < 1$ in the expression of the Lagrangian \cite{BK02}.
We will consider the coupling constant in the interval $\frac{2}{3}<\nu < \frac45$,
and differently from the sine-Gordon model,
we do not consider solitons and interpret that there are only two breathers $b_1, b_2$,
by taking the \emph{maximal analyticity} (see below) in a strict
sense.\footnote{In the form factor programme \cite{BFKZ99}, for a given $0 < \nu < 1$, there are $K$ breathers,
where $K$ is the largest integer such that $K\nu < 1$.
Especially, if $\frac12 < \nu <1$, there is only one breather $b_1$, differently from our case
(we are indeed not considering the Thirring model).
}
The masses of the breathers are given by $m_{b_\ell} = 2m \sin\frac{\ell\nu\pi}{2}$,
where $m > 0$ and $\ell = 1,2$. These particles are neutral and hence the charge conjugate of $b_\ell$ (denoted with $\bar b_\ell$ in literature)
is $b_\ell$ itself.

In this case, the elastic two-particle scattering processes are characterized by a matrix-valued function with only non-zero components $S^{b_1 b_1}_{b_1 b_1}(\theta)$, $S^{b_1 b_2}_{b_2 b_1}(\theta)$, $S^{b_2 b_1}_{b_1 b_2}(\theta)$ and $S^{b_2 b_2}_{b_2 b_2}(\theta)$, where $\theta$ is the difference of the rapidities of the  incoming particles.
We will give explicit expressions for them in Section \ref{scattering}. They are the breather-breather S-matrix components of the sine-Gordon model
multiplied by so-called CDD factors.

The particles $b_1, b_2$ may form a bound state in a scattering process.
We declare that the possible \emph{fusion processes} are only of three types, $(b_1 b_1) \rightarrow b_2$, $(b_1 b_2) \rightarrow b_1$  and  $(b_2 b_1) \rightarrow b_1$. On the other hand, $(b_2 b_2)$ is not a fusion. 
The corresponding imaginary rapidities of the fusing particles are denoted by $\theta_{(b_1 b_1)}^{b_2}$ for the first fusion, and $\theta_{( b_1 b_2)}^{b_1}$, $\theta_{(b_2 b_1)}^{b_1}$ for the second two types of fusion. Correspondingly, we do not specify the rapidity $\theta_{(b_2b_2)}$, since there is no fusion $(b_2 b_2)$.
The actual values will be given in Section \ref{scattering}.

In the same way as in \cite{CT16-diag}, to these fusion processes there correspond the so-called \emph{fusion angles},
which determine the position of the simple poles in the components $S^{b_1 b_1}_{b_1 b_1}(\zeta)$, $S^{b_2 b_1}_{b_1 b_2}(\zeta)$ and $S^{b_1 b_2}_{b_2 b_1}(\zeta)$ in the \emph{physical strip}
$0< \im \zeta < \pi$.  Specifically, for the fusion $(b_1 b_1) \rightarrow b_2$, $S^{b_1 b_1}_{b_1 b_1}(\zeta)$
has a simple pole at $\zeta = i\theta_{b_1 b_1}^{b_2}$, where 
 \begin{align*}
 \theta_{b_1 b_1}^{b_2} := \theta_{(b_1 b_1)}^{b_2} + \theta_{(b_1 b_1)}^{b_2} \;(=2\theta_{(b_1 b_1)}^{b_2}).
 \end{align*}
Similarly, $S^{b_2 b_1}_{b_1 b_2}(\zeta)$, corresponding to the fusion process $(b_2 b_1) \rightarrow b_1$, has a simple pole at $\zeta = i\theta_{b_1 b_2}^{b_1}$, where 
 \begin{align*}
 \theta_{b_1 b_2}^{b_1} := \theta_{(b_1 b_2)}^{b_1} + \theta_{(b_2 b_1)}^{b_1},
 \end{align*}
and the same holds for the S-matrix component  $S^{b_1 b_2}_{b_2 b_1}(\zeta)$. In our construction,
the poles in the component $S^{b_2 b_2}_{b_2 b_2}(\zeta)$ do not matter.
We will indeed introduce the additional concept of \emph{elementary particle} in analogy with \cite{CT16-diag},
and we assume the so-called ``maximal analyticity'' only for the elementary particle $b_1$.

These angles correspond to \textbf{$s$-channel poles} and in the model under investigation they are explicitly given in Table~\ref{Table}.
The S-matrix components $S^{b_1b_1}_{b_1 b_1}$, $S^{b_2 b_1}_{b_1 b_2}$, $S^{b_1b_2}_{b_2 b_1}$ and $S^{b_2 b_2}_{b_2 b_2}$ are meromorphic functions on $\CC$, which we present below.
In addition, we will introduce the bound state intertwiners
$\eta^{b_2}_{b_1 b_1}$, $\eta^{b_1}_{b_2 b_1}$ and $\eta^{b_1}_{b_1 b_2}$ (there is no corresponding matrix element for $(b_2 b_2)$, as this is not a fusion.)
In a general non-diagonal case, they formally diagonalize the S-matrix components above at the corresponding pole,
and their eigenvalues correspond to the residues.
They were also defined in \cite{BFKZ99} and more explicitly in \cite[before Eq.(1.13)]{Quella99}
and here we adopt a slightly modified convention, as below.

\subsection{Scattering data}\label{scattering}
The input which specifies the S-matrix of our model is the following.

\begin{itemize}
 \item The {\bf coupling constant} $\nu$, which is a parameter such that $\frac{2}{3}< \nu < \frac45$
 and the {\bf mass parameter $m > 0$} which determines the masses of the breathers (see below).
 For the value of $\nu$ in the range above, we consider two breathers, $b_1, b_2$.
 Indeed, $K=2$ is the largest integer such that $K\nu < 2$. 
 \item The S-matrix components:
 $S_{b_k b_\ell}^{b_\ell b_k}(\zeta) =S_{\text{SG}}\phantom{}_{b_k b_\ell}^{b_\ell b_k}(\zeta) S_{\text{CDD}}\phantom{}_{b_k b_\ell}^{b_\ell b_k}(\zeta)$,
where
\begin{align*}
S_{\text{SG}}\phantom{}^{b_1 b_1}_{b_1 b_1}(\zeta) &= \frac{\tan \frac{1}{2i}\left(  \zeta + i\pi \nu \right)}{\tan \frac{1}{2i} \left( \zeta -i\pi \nu  \right)}, \\ 
S_{\text{SG}}\phantom{}^{b_2 b_1}_{b_1 b_2}(\zeta) =S_{\text{SG}}\phantom{}^{b_1 b_2}_{b_2 b_1}(\zeta) &= \frac{\tan \frac{1}{2i}\left(  \zeta + \frac{3i\pi \nu}2 \right)}{\tan \frac{1}{2i} \left( \zeta -\frac{3i\pi \nu}2  \right)}\frac{\tan \frac{1}{2i}\left(  \zeta + \frac{i\pi \nu}2 \right)}{\tan \frac{1}{2i} \left( \zeta -\frac{i\pi \nu}2  \right)}, \\
S_{\text{SG}}\phantom{}^{b_2 b_2}_{b_2 b_2}(\zeta) &= \frac{\tan \frac{1}{2i}\left(  \zeta + 2i\pi \nu \right)}{\tan \frac{1}{2i} \left( \zeta -2i\pi \nu  \right)}\frac{\left( \tan \frac{1}{2i}\left(  \zeta + i\pi \nu \right) \right)^2}{\left( \tan \frac{1}{2i} \left( \zeta -i\pi \nu  \right) \right)^2}. 
\end{align*}
are the breather-breather S-matrix components of the sine-Gordon model (see \cite{BK02, Quella99}),
and $S_{\text{CDD}}\phantom{}_{b_k b_\ell}^{b_\ell b_k}$ are introduced as follows:
\begin{align}\label{Sb1}
 S_{\text{CDD}}\phantom{}_{b_1 b_1}^{b_1 b_1}(\zeta) &:=
 \frac{\sinh\frac12\left(\theta - i\pi(\nu - \nu_-)\right)}{\sinh\frac12\left(\theta + i\pi(\nu - \nu_-)\right)}
 \cdot \frac{\sinh\frac12\left(\theta - i\pi(\nu + \nu_+)\right)}{\sinh\frac12\left(\theta + i\pi(\nu + \nu_+)\right)} \nonumber \\
 &\qquad \times \frac{\sinh\frac12\left(\theta - i\pi(1- \nu + \nu_-)\right)}{\sinh\frac12\left(\theta + i\pi(1- \nu + \nu_-)\right)}
 \cdot \frac{\sinh\frac12\left(\theta - i\pi(1- \nu - \nu_+)\right)}{\sinh\frac12\left(\theta + i\pi(1- \nu - \nu_+)\right)},
\end{align}
and expecting the bootstrap equation (see condition  \ref{bootstrap} below), we also define
\begin{align}\label{Sb2}
 S_{\text{CDD}}\phantom{}_{b_1 b_2}^{b_2 b_1}(\zeta) &= S_{\text{CDD}}\phantom{}_{b_2 b_1}^{b_1 b_2}(\zeta) \nonumber \\
 &:= S_{\text{CDD}}\phantom{}_{b_1 b_1}^{b_1 b_1}(\zeta + i\theta_{(b_1b_1)}^{b_2})S_{\text{CDD}}\phantom{}_{b_1 b_1}^{b_1 b_1}(\zeta - i\theta_{(b_1b_1)}^{b_2})\nonumber \\
 &=\frac{\sinh\frac12\left(\theta - i\pi(\frac32\nu - \nu_-)\right)}{\sinh\frac12\left(\theta + i\pi(\frac12\nu - \nu_-)\right)}
 \cdot \frac{\sinh\frac12\left(\theta - i\pi(\frac32\nu + \nu_+)\right)}{\sinh\frac12\left(\theta + i\pi(\frac12\nu + \nu_+)\right)}\nonumber \\
 &\qquad \times \frac{\sinh\frac12\left(\theta - i\pi(1- \frac12\nu + \nu_-)\right)}{\sinh\frac12\left(\theta + i\pi(1- \frac32\nu + \nu_-)\right)}
 \cdot \frac{\sinh\frac12\left(\theta - i\pi(1- \frac12\nu - \nu_+)\right)}{\sinh\frac12\left(\theta + i\pi(1- \frac32\nu - \nu_+)\right)} \nonumber \\
 &\qquad \times \frac{\sinh\frac12\left(\theta - i\pi(\frac12\nu - \nu_-)\right)}{\sinh\frac12\left(\theta + i\pi(\frac32\nu - \nu_-)\right)}
 \cdot \frac{\sinh\frac12\left(\theta - i\pi(\frac12\nu + \nu_+)\right)}{\sinh\frac12\left(\theta + i\pi(\frac32\nu + \nu_+)\right)} \nonumber \\
 &\qquad \times \frac{\sinh\frac12\left(\theta - i\pi(1- \frac32\nu + \nu_-)\right)}{\sinh\frac12\left(\theta + i\pi(1- \frac12\nu + \nu_-)\right)}
 \cdot \frac{\sinh\frac12\left(\theta - i\pi(1- \frac32\nu - \nu_+)\right)}{\sinh\frac12\left(\theta + i\pi(1- \frac12\nu - \nu_+)\right)},  \\
  S_{\text{CDD}}\phantom{}_{b_2 b_2}^{b_2 b_2}(\zeta)
 &:= S_{\text{CDD}}\phantom{}_{b_1 b_2}^{b_2 b_1}(\zeta + i\theta_{(b_1b_1)}^{b_2})S_{\text{CDD}}\phantom{}_{b_1 b_2}^{b_2 b_1}(\zeta - i\theta_{(b_1b_1)}^{b_2}) \nonumber \\
 &= S_{\text{CDD}}\phantom{}_{b_1 b_1}^{b_1 b_1}(\zeta + i\theta_{b_1b_1}^{b_2})S_{\text{CDD}}\phantom{}_{b_1 b_1}^{b_1 b_1}(\zeta)^2S_{\text{CDD}}\phantom{}_{b_1 b_1}^{b_1 b_1}(\zeta - i\theta_{b_1b_1}^{b_2}).\nonumber
\end{align}
We do not need an explicit expression for $S_{\text{CDD}}\phantom{}_{b_2 b_2}^{b_2 b_2}$, and we omit computing it.
$\nu_-$ and $\nu_+$ are parameters satisfying the following set of conditions:
\begin{enumerate}[(i)]
 \item $ \nu_-, \nu_+ >0 $.\label{i}
 \item $\nu_- \in \left( \frac{3}{2}\nu -1, \frac{1}{2}\nu  \right)$.\label{ii}
 \item $\nu_+ \in (0, 1-\nu)$.\label{iii}
 \item $1-\nu = \nu_-+\nu_+$.\label{iv}
\end{enumerate}
For $\frac23 < \nu < \frac45$, there are such $\nu_-,\nu_+$.
Indeed, by rewriting every condition \eqref{i}--\eqref{iii} only in terms of $\nu_+$
through \eqref{iv} which is equivalent to $\nu_- = 1-\nu -\nu_+$,
we obtain $0 < \nu_+ < 1-\nu$ and $1-\frac32 \nu < \nu_+ < 2-\frac52\nu$,
which have always a nontrivial intersection for $\frac23 < \nu < \frac45$
(on the other hand, in the interval $\frac45 < \nu < 1$ there is no such intersection).

Let us take such $\nu_-, \nu_+$.
\eqref{iv} is equivalent to $-\left(\frac{1}{2}\nu - \nu_- \right) = 1-\frac{3}{2}\nu - \nu_+$,
therefore, from \eqref{Sb2} and $\sinh \frac12(\zeta + 2\pi i) = -\sinh\frac12\zeta$ we have
\begin{align*}
 S_{\text{CDD}}\phantom{}_{b_1 b_2}^{b_2 b_1}(\zeta) &= S_{\text{CDD}}\phantom{}_{b_2 b_1}^{b_1 b_2}(\zeta) \nonumber \\
&=
 \frac{\sinh\frac12\left(\theta - i\pi(\frac32\nu - \nu_-\right))}{\sinh\frac12\left(\theta + i\pi(\frac12\nu - \nu_-\right))}
 \cdot \frac{\sinh\frac12\left(\theta - i\pi(\frac32\nu + \nu_+\right))}{\sinh\frac12\left(\theta + i\pi(\frac12\nu + \nu_+\right))} \nonumber \\
 &\qquad \times \frac{-1}{\sinh\frac12\left(\theta + i\pi(1- \frac32\nu + \nu_-\right))}
 \cdot \frac{\sinh\frac12\left(\theta - i\pi(1- \frac12\nu - \nu_+\right))}{1} \nonumber \\
 &\qquad \times \frac{1}{\sinh\frac12\left(\theta + i\pi(\frac32\nu - \nu_-\right))}
 \cdot \frac{\sinh\frac12\left(\theta - i\pi(\frac12\nu + \nu_+\right))}{1} \nonumber\\
 &\qquad \times \frac{\sinh\frac12\left(\theta - i\pi(1- \frac32\nu + \nu_-\right))}{\sinh\frac12\left(\theta + i\pi(1- \frac12\nu + \nu_-\right))}
 \cdot \frac{\sinh\frac12\left(\theta - i\pi(1- \frac32\nu - \nu_+\right))}{\sinh\frac12\left(\theta + i\pi(1- \frac12\nu - \nu_+\right))}, \\
\end{align*}
and it is straightforward to see that these S-matrix components have no pole in the physical strip
$0 < \im \zeta < \pi$.

 \item There are only three possible fusion processes $(b_1 b_1) \rightarrow b_2$,  $(b_2 b_1) \rightarrow b_1$ and $(b_1 b_2) \rightarrow b_1$. Note that $(b_2 b_2)$ is not a fusion.  The corresponding rapidities of particles $\theta_{(b_1 b_1)}^{b_2}$, $\theta_{(b_1 b_2)}^{b_1}$ and $\theta_{(b_2 b_1)}^{b_1}$ are presented in the fusion table (Table \ref{Table}).
We define the fusion angles by $\theta_{\a \b}^{\g} := \theta_{(\a \b)}^{\g} +\theta_{(\b \a)}^{\g}$
if $(\a\b) \to \g$ is a two fusion process, where $\a,\b,\g = b_1$ or $b_2$.

\begin{table}
\begin{tabular}{|c|c|c|}
 \hline
 processes & rapidities of particles & fusion angles \\
 \hline
 $(b_1 b_1) \longrightarrow b_2 $ & $\theta_{(b_1 b_1)}^{b_2} = \frac{\pi \nu }2$ & $\theta_{b_1 b_1}^{b_2} = \pi \nu$ \\
 \hline
 $(b_2 b_1) \longrightarrow b_1, (b_1 b_2) \longrightarrow b_1$ &
$\theta_{(b_1 b_2)}^{b_1} =\pi(1 - \nu), \theta_{(b_2 b_1)}^{b_1} = \frac{\pi \nu}2$ & $\theta_{b_2 b_1}^{b_1} = \theta_{b_1 b_2}^{b_1} = \pi \left( 1-\frac{\nu}{2} \right)$ \\
 \hline
 $(b_2 b_2)$ not a fusion & &   \\
 \hline
\end{tabular}
\caption{Fusions and angles}\label{Table}
\end{table}

The data collected above satisfy the following ``axioms''
(in general, these axioms involve the charge conjugation, but for breathers it is trivial, $\bar b_1 = b_1$ and $\bar b_2 = b_2$).
In the following, $k, \ell=1,2$.
\begin{enumerate}
  \renewcommand{\theenumi}{(S\arabic{enumi})}
  \renewcommand{\labelenumi}{\theenumi}
 \item \label{meromorphic} {\bf Meromorphy.} The functions $S^{b_\ell b_k}_{b_k b_\ell}(\zeta)$ are meromorphic on $\mathbb{C}$.
 \item \label{parity} {\bf Parity symmetry.} $S^{b_\ell b_k}_{b_k b_\ell}(\zeta) = S^{b_k b_\ell}_{b_\ell b_k}(\zeta)$.
 \item \label{unitarity} {\bf Unitarity.} $S^{b_\ell b_k}_{b_k b_\ell}(\zeta)^{-1} = \overline{S^{b_k b_\ell}_{b_\ell b_k}(\bar\zeta)}$.
 \item \label{hermitian} {\bf Hermitian analyticity.} $S^{b_\ell b_k}_{b_k b_\ell}(\zeta) = S^{b_\ell b_k}_{b_k b_\ell}(-\zeta)^{-1}$.
 \item \label{crossing} {\bf Crossing symmetry.} $S^{b_\ell b_k}_{b_k b_\ell}(i\pi -\zeta) = S^{b_k b_\ell}_{b_\ell b_k}(\zeta)$.
 \item \label{bootstrap} {\bf Bootstrap equation.} 
 Let $\a,\b,\g,\mu = b_1$ or $b_2$.
 If $(\alpha \beta) \rightarrow \gamma$ is a fusion process in Table \ref{Table},
 there holds
 \begin{equation*}
 S^{\gamma \mu}_{\mu \gamma}(\zeta) = S^{\alpha \mu}_{\mu \alpha}(\zeta +i\theta_{(\alpha \beta)}^\gamma) S^{\beta \mu}_{\mu \beta}(\zeta -i\theta_{(\beta \alpha)}^\gamma).
 \end{equation*}
  
 \item \label{atzero} {\bf Value at zero.} $S^{b_k b_k}_{b_k b_k}(0) = -1$.
 \item \label{regularity}  \textbf{Regularity.} The components $S^{b_\ell b_k}_{b_k b_\ell}$ have only finitely many zeros in the physical strip and
 there is $\k > 0$ such that $\|S\|_\k := \sup\left\{|S_{b_k b_\ell}^{b_\ell b_k}(\zeta)|: \zeta \in \RR + i(-\k,\k) \right\} < \infty$ (the value of $\k$ depends on
 the parameters $\nu, \nu_-, \nu_+$).
 \item \label{maximal} \textbf{Maximal analyticity (for $b_1$).}\footnote{We call this ``maximal analyticity''
 because each $s$-channel pole at $i\theta_{b_1 b_k}^{b_\ell}$ has a corresponding
 entry $(b_1 b_k)\to b_\ell$ in the fusion Table \ref{Table}.
 It should be noted that this is required only for the S-matrix components containing $b_1$, the ``elementary particle'' defined below.}
 The component $S^{b_1 b_1}_{b_1 b_1}(\zeta)$ has only two simple poles in the physical strip. They are at
 $i\theta_{b_1 b_1}^{b_2} = i\pi \nu$ (called $s$-channel pole) and $i\theta_{b_1 b_1}^{\prime b_2} := i\pi - i\theta^{b_2}_{b_1 b_1} = i\pi(1 - \nu)$
(called $t$-channel pole, whose existence follows from crossing symmetry). Similarly, the component $S^{b_2 b_1}_{b_1 b_2}(\zeta)$ has also only
two simple poles, i.e. an $s$-channel pole at $i\theta_{b_1 b_2}^{b_1} = i\pi \left( 1- \frac{ \nu}{2}\right)$ and a $t$-channel pole at $i\theta_{b_1 b_2}^{\prime b_1} := i\pi - i\theta^{b_1}_{b_2 b_1} = \frac{i\pi \nu}{2}$.
 
 Furthermore, $S^{b_k b_1}_{b_1 b_k}$ have no double or higher poles in the physical strip, $k=1,2$.
 \item \label{positiveresidue} {\bf Positive residue (for $b_1)$.}
 If $(b_1 b_k) \to b_\ell$ is a fusion process, then
 \[
  R_{b_1 b_k}^{b_\ell} := \res_{\zeta = i\theta_{b_1 b_k}^{b_\ell}} S_{b_1 b_k}^{b_k b_1}(\zeta) \in i\RR_+.
 \]
 
\end{enumerate}

\end{itemize}

\paragraph{Proof of the axioms.}
\begin{itemize}
\item \ref{meromorphic}--\ref{bootstrap} and \ref{regularity}.
These properties are already satisfied by the S-matrix with components ${S_{\mathrm{SG}}}_{b_k b_\ell}^{b_\ell b_k}$
of the sine-Gordon model (and well-known in the literature).
It is also straightforward to check that $S_{\text{CDD}}\phantom{}_{b_k b_\ell}^{b_\ell b_k}(\zeta)$ satisfy
\ref{meromorphic}--\ref{crossing} and \ref{regularity}.
As for \ref{bootstrap}, we have by construction
\[
 S_{\text{CDD}}\phantom{}_{b_1 b_2}^{b_2 b_1}(\zeta) = S_{\text{CDD}}\phantom{}_{b_1 b_1}^{b_1 b_1}\left(\zeta + \frac{i\pi \nu}2\right) S_{\text{CDD}}\phantom{}_{b_1 b_1}^{b_1 b_1}\left(\zeta - \frac{i\pi \nu}2\right).
\]
By the properties mentioned above (in particular, hermitian analyticity), we have
\begin{align*}
 S_{\text{CDD}}\phantom{}_{b_1 b_1}^{b_1 b_1}(\zeta) = S_{\text{CDD}}\phantom{}_{b_1 b_2}^{b_2 b_1}\left(\zeta + \frac{i\pi \nu}2\right) S_{\text{CDD}}\phantom{}_{b_1 b_1}^{b_1 b_1}\left(\zeta - i\pi (1 - \nu)\right)
\end{align*}
Similarly, the bootstrap for $b_2$ can be satisfied by construction.

Therefore, the products $S\phantom{}_{b_k b_\ell}^{b_\ell b_k}(\zeta) = S_{\text{SG}}\phantom{}_{b_k b_\ell}^{b_\ell b_k}(\zeta)S_{\text{CDD}}\phantom{}_{b_k b_\ell}^{b_\ell b_k}(\zeta)$
satisfy them as well.
\item \ref{atzero}. It is easy to see that
$S_{\text{SG}}\phantom{}_{b_k b_k}^{b_k b_k}(0) = -1$, while $S_{\text{CDD}}\phantom{}_{b_k b_k}^{b_k b_k}(0) = 1$,
therefore, we have $S_{\text{CDD}}\phantom{}_{b_k b_k}^{b_k b_k}(0) = -1$.
\item \ref{maximal}. The expression of $S_{\text{CDD}}\phantom{}_{b_1 b_k}^{b_k b_1 }(\zeta)$ does not have poles in the physical strip,
so the pole structure of $S_{b_1 b_k}^{b_k b_1}(\zeta)$ is determined by $S_{\text{SG}}\phantom{}_{b_1 b_k}^{b_k b_1 }(\zeta)$,
which is easy to check (and known in the literature).
\item \ref{positiveresidue} is violated in the sine-Gordon model,
indeed 
\[
R_{\text{SG}}\phantom{}^{b_2}_{b_1 b_1} := \res_{\zeta = i\theta_{b_1 b_1}^{b_2}} {S_\text{SG}}_{b_1 b_1}^{b_1 b_1}(\zeta) = 2i \tan(\pi \nu) \in -i\RR_+,
\]
for our range of $\nu \in (\frac23, \frac45)$.

On the other hand, by counting the zeros on the imaginary line
and by recalling that $S_{\text{CDD}}\phantom{}_{b_1 b_1}^{b_1 b_1}(0) = 1$,
one can see that $S_{\text{CDD}}\phantom{}_{b_1 b_1}^{b_1 b_1}(i\pi \nu) < 0$,
hence we obtain $R_{b_1 b_1}^{b_2} = \res_{\zeta = i\pi \nu} S_{b_1 b_1}^{b_1 b_1}(\zeta) \in i\RR_+$ as desired.
From this it follows that $R_{b_1 b_2}^{b_1} \in i\RR_+$ as well, since we will see below that $R_{b_1 b_2}^{b_1} = R_{b_1 b_1}^{b_2}$.
\end{itemize}

The residues of $S^{b_k b_1}_{b_1 b_k}(\zeta)$ will play an important role,
so we give them symbols.
\begin{align*}
 R^{b_\ell}_{b_1 b_k} := \operatorname{Res}_{\zeta =i\theta_{b_1 b_k}^{b_\ell}} S_{b_1 b_k}^{b_k b_1}(\zeta),
\quad R^{\prime b_\ell}_{b_1 b_k} := \operatorname{Res}_{\zeta =i\theta_{b_1 b_k}^{\prime b_\ell}} S_{b_1 b_k}^{b_k b_1}(\zeta) \\
 R^{b_\ell}_{b_k b_1} := \operatorname{Res}_{\zeta =i\theta_{b_k b_1}^{b_\ell}} S_{b_k b_1}^{b_1 b_k}(\zeta),
\quad R^{\prime b_\ell}_{b_k b_1} := \operatorname{Res}_{\zeta =i\theta_{b_k b_1}^{\prime b_\ell}} S_{b_k b_1}^{b_1 b_k}(\zeta)
 \end{align*}
and it follows that $R^{b_\ell}_{b_1 b_k} = R^{b_\ell}_{b_k b_1}$.

As before, we also introduce\footnote{We use a slightly different convention from \cite{Quella99}: For a fusion process $(\alpha \beta) \rightarrow \gamma$, we have $\eta^{\g}_{\a \b} = \sqrt{2\pi}\,\eta^{\gamma}_{\alpha \beta}\phantom{}^{(\text{Quella})}$.}
the symbols $\eta^{b_2}_{b_1 b_1}$ and $\eta^{b_1}_{b_2 b_1}$ by the following formula:
\begin{equation}\label{Reseta}
\eta^{b_2}_{b_1b_1} = i\sqrt{2 \pi \left|R^{b_2}_{b_1b_1}\right|},
\quad  \eta^{b_1}_{b_1 b_2} = i\sqrt{2 \pi \left|R^{b_1}_{b_1 b_2}\right|},
\quad  \eta^{b_1}_{b_2 b_1} = i\sqrt{2 \pi \left|R^{b_1}_{b_2 b_1}\right|}.
\end{equation}
Furthermore, by convention, we set to zero any residues and matrix elements of the above type which do \textit{not} correspond to a fusion in Table \ref{Table}.
From the properties \ref{parity}--\ref{atzero} of the S-matrix,
there is a number of other properties of the fusion angles and of the residues that follow,
and we refer for the proofs to \cite[Sec.\! 2.1]{CT16-diag}. We would mention here only the following.
The residue of the \emph{t-channel} pole is related to the residue of the s-channel pole by $R^{\prime b_2}_{b_1 b_1} = -  R_{ b_1 b_1}^{b_2}$ and $R^{\prime b_1}_{b_2 b_1} = -  R_{ b_1 b_2}^{b_1}$, and that by \ref{parity}, $R_{b_1 b_2}^{b_1} = R_{b_2 b_1}^{b_1}$. \ref{bootstrap} and \ref{atzero} imply that $R_{b_1 b_1}^{b_2} = R_{b_1 b_2}^{b_1}$.
Furthermore, if $(b_1 b_k) \to b_\ell$ is a fusion process, the fusion angles are also related by
\begin{equation}\label{prop:angles}
\pi - \theta^{b_\ell}_{b_1 b_k} = \theta^{b_k}_{( b_1 b_\ell)}, \quad  \theta^{b_k}_{(b_\ell b_1)} = \theta^{b_\ell}_{(b_k b_1)}.
\end{equation}

From the equality $R_{b_1 b_1}^{b_2} = R_{b_1 b_2}^{b_1}$ and the parity $R_{b_1 b_2}^{b_1} = R_{b_2 b_1}^{b_1}$, 
it also holds that
$\eta^{b_2}_{b_1 b_1} =\eta^{ b_1}_{b_1 b_2} = \eta^{ b_1}_{ b_2 b_1}$. 

\paragraph{Particle spectrum.}
Given the mass parameter $m > 0$, we define the masses of the particles as
\[
 m_{b_1} = 2m \sin\frac{\nu\pi}{2}, \quad
 m_{b_2} = 2m \sin\frac{2\nu\pi}{2}.
\]
They satisfy the following ``fusion'' rule:
\begin{align}\label{eq:massfusion}
 m_{b_2} = m_{b_1} \cos\theta^{b_2}_{(b_1 b_1)} + m_{b_1} \cos\theta^{b_2}_{(b_1 b_1)}, \quad
 m_{b_1} = m_{b_1} \cos\theta^{b_1}_{(b_1 b_2)} + m_{b_2} \cos\theta^{b_1}_{(b_2 b_1)}.
\end{align}

As $b_1$ plays a special role in our methods,
we call it an {\bf elementary particle} as in \cite[Sec.\! 2.1]{CT16-diag}.

\section{The physical Hilbert space}\label{Fock}

From the scattering data of Section \ref{scattering}, we construct basic mathematical structures
for the wedge-observables in the quantum field theory on the $S$-symmetric Fock space.
The construction can be thought of as a kind of deformation of a free field theory with the input given by the S-matrix.
The single-particle Hilbert space accommodates the two species of particles:
\begin{align*}
\mathcal{H}_1 =  \displaystyle{\bigoplus_{k=1,2}}  \mathcal{H}_{1,b_k}, \quad \mathcal{H}_{1,b_k}= L^2(\mathbb{R},d\theta).
\end{align*}
An element $\Psi_1 \in \mathcal{H}_1$ can be identified as a vector valued function with components $\theta \mapsto \Psi_1^{b_k}(\theta)$. 
On the unsymmetrized $n$-particle space $\H_1^{\otimes n}$, there is a unitary representation
$D_n$ of the symmetric group $\mathfrak{G}_n$ which, with $\pmb{\theta} := (\theta_1,\cdots, \theta_n)$, acts as
\begin{equation*}
  (D_n(\tau_j)\Psi_n)^{\pmb{b_k}}(\pmb{\theta}) = S_{ b_{k_{j+1}}  b_{k_j} }^{ b_{k_j} b_{k_{j+1}}}(\theta_{j+1}-\theta_j)\Psi^{b_{k_1}\cdots b_{k_{j+1}} b_{k_j}\cdots b_{k_n}}_n(\theta_1,\cdots, \theta_{j+1},\theta_j,\cdots,\theta_n),
\end{equation*}
where $k_1, \ldots, k_n \in \{ 1,2 \}$, $\pmb{\theta} := (\theta_1 , \ldots , \theta_n)$, $\pmb{b_k} := (b_{k_1}, \ldots, b_{k_n} )$ and $\tau_j \in \mathfrak{G}_n$ is the transposition $(j, j+1) \rightarrow (j+1, j)$.

The full Hilbert space $\mathcal{H}$ is
$\mathcal{H} := \bigoplus_{n=0}^\infty \mathcal{H}_n$ with $\mathcal{H}_0 =\mathbb{C}\Omega$,
where $\mathcal{H}_n = P_n \mathcal{H}_1 ^{\otimes n}$ and $P_n := \frac{1}{n!}\sum_{\sigma \in \mathfrak{G}_n} D_n(\sigma)$ is an orthogonal projection.
The elements of $\mathcal{H}$ are $L^2$-sequences $\Psi = (\Psi_0, \Psi_1, \ldots)$, where $\Psi_n$ are \emph{$S$-symmetric functions}, namely invariant under the action of $\mathfrak{G}_n$. 
Finally, we denote by $\mathcal{D}$ the linear hull (without closure) of $\{\H_n\}$.  

There is a unitary representation $U$ of the proper orthochronous Poincar\'e group $\poincare$ on $\H$
which preserves each $\H_n$,
\begin{align*}
U := \bigoplus_n U_n, \quad (U(a,\lambda)\Psi)^{\pmb{b_k}}_n(\pmb{\theta}) := \exp \left(i \sum_{l=1}^n p_{b_{k_l}}(\theta_l) \cdot a\right) \Psi^{\pmb{b_k}}_n(\theta_1 - \l, \cdots, \theta_n - \l),
\end{align*}
where $p_{b_{k_l}}(\theta)= (m_{b_{k_l}}\cosh \theta, m_{b_{k_l}}\sinh \theta)$.
Additionally, there is an antiunitary representation of the CPT operator on $\mathcal{H}$:
\begin{equation*}
J := \bigoplus_n J_n, \quad (J \Psi)_n^{\pmb{b_k}}(\pmb{\theta}) :=  \overline{\Psi_n^{b_{k_n} \ldots b_{k_1}} ( \theta_n, \ldots,   \theta_1)}.
\end{equation*}
We consider test functions with multi-components and are chosen as $g \in \bigoplus_{k=1}^2 \mathscr{S}(\mathbb{R}^2)$ with $g_{b_k} \in \mathscr{S}(\mathbb{R}^2)$, and we adopt the following convention:
\begin{equation*}
g^{\pm}_{b_k}(\theta) := \frac{1}{2\pi} \int d^2 x\, g_{b_k} (x) e^{\pm i p_{b_k} (\theta)\cdot x}.
\end{equation*}
We note that\footnote{Our convention of the Lorentz metric is $a \cdot b = a_0b_0-a_1b_1$.}
if $g_{b_k}$ is supported in $W_\R$, then $g^+_{b_k}(\theta)$ has a bounded analytic continuation in $\RR +i(-\pi, 0)$ and $|g^+_{b_k}(\theta +i\lambda)|$ decays rapidly as $\theta \rightarrow \pm \infty$ in the strip for $\l \in (-\pi,0)$. Moreover, $g^+_{b_k}(\theta -i\pi)= g^-_{b_k}(\theta)$.

There is a natural action of the proper Poincar\'e group on $\RR^2$ and on the space of test functions,
denoted by $g_{(a,\lambda)}$,
and it is compatible with the action on the one-particle space:
\[
 (g_{(a, \lambda)})^\pm_{b_k} = U_1(a,\lambda)g^\pm_{b_k}.
\]

The CPT transformation acts also on multi-components test functions, which we denote by $j$,
as $g \mapsto g_j$, $(g_j)_{b_k} (x) := \overline{g_{ b_k }(-x)}$,
and this is again compatible with $J_1$:
$(g_j)^\pm_{b_k} (\theta)= J_1g^\pm_{b_k} (\theta) = \overline{g^{\pm}_{b_k}(\theta)}$.

Moreover, we introduce the complex conjugate of a multi-component test function by
$(g^*)_{b_k} (x) :=\overline{g_{b_k}(x)}$ and if $g=g^*$, then 
we say that $g$ is real and it follows that $\overline{g^\pm_{b_k}(\overline \zeta)} = g^\mp_{b_k}(\zeta)$
(c.f.\! \cite[Proposition 3.1]{LS14}).

\subsubsection*{Zamolodchikov-Faddeev algebra}

Similarly to \cite{LS14}, creators and annihilators $z^\dagger_{b_k} (\theta), z_{b_k} (\theta)$ are introduced
in the $S$-symmetric Fock space $\mathcal{H}$. For $\varphi \in \H_1$, their actions on vectors $\Psi = (\Psi_n) \in \D$ are given by
\begin{align*}
(z(\varphi)\Psi)^{\pmb{b_k}}_n (\pmb{\theta}) &= \sqrt{n+1}\sum_{l=1,2} \int d\theta' \overline{\varphi^{b_l}(\theta')}\Psi^{b_l \pmb{b_k}}_{n+1}(\theta',\pmb{\theta}),\\
z^\dagger(\varphi) &= (z(\varphi))^*
\end{align*}
(see \cite[Proposition 2.4]{LS14}) and they formally fulfill the following Zamolodchikov-Faddeev algebra:
\begin{align*}
z^\dagger_{b_k}(\theta)z^\dagger_{b_l}(\theta') &= S^{b_l b_k}_{b_k b_l}(\theta -\theta')z^\dagger_{b_l} (\theta')z^\dagger_{b_k}(\theta),\\
z_{b_k}(\theta)z_{b_l}(\theta') &= S^{b_l b_k}_{b_k b_l}(\theta -\theta')z_{b_l}(\theta')z_{b_k}(\theta),\\
z_{b_k}(\theta)z^\dagger_{b_l}(\theta') &= S^{b_k b_l}_{b_l b_k}(\theta' -\theta)z^\dagger_{b_l} (\theta')z_{b_k} (\theta)+\delta^{b_k b_l}\delta(\theta -\theta')\1_{\mathcal{H}}.
\end{align*}
They are opereator-valued distributions defined on $\D$ and bounded on each $n$-particle space $\H_n$
when smeared by a test function. 

Let $f \in \bigoplus _{k=1,2}\mathscr{S}(\mathbb{R}^2) $, we define
\begin{align*}
\phi(f)&:= z^\dagger(f^+)+ z(J_1f^-) \\
       &\left(= \sum_{k=1,2} \int d\theta\, \left( f^+_{b_k}(\theta) z^\dagger_{b_k}(\theta)+ (J_1f^-)_{b_k} (\theta)z_{b_k}(\theta) \right)\right).
\end{align*}
This multi-component quantum field\footnote{If the S-matrix $S(\zeta)$ were \emph{analytic in the physical strip},
$\phi(f)$ could be considered as an observable localized in the standard left wedge $W_{\mathrm L}$
and if furthermore $S$ is diagonal with additional regularity conditions,
one would be able to obtain a Haag-Kastler net with minimal length \cite{LS14, AL16}.
In contrast, our S-matrix has poles in the physical strip.}
is defined on the subspace $\D$ of $\mathcal{H}$ of vectors with finite particle number and
the properties listed in \cite[Proposition 3.1]{LS14} are fulfilled, as long as the analyticity
in the physical strip is not used.
We also introduce $\phi'$, the reflected field defined for $g \in \mathscr{S}(\mathbb{R}^2)$,
\begin{equation*}
\phi'(g):= J \phi(g_j)J = z'^\dagger (g^+) + z'(J_1 g^-),
\end{equation*}
where $z', z'^{\dagger}$ are the reflected creators and annihilators $z'_{b_k} (\theta):= J z_{b_k}(\theta)J$
and $z'^\dagger_{b_k} (\theta):= J z_{b_k}^\dagger(\theta)J$.

For the class of two-particle S-matrices $S(\theta)$ with components which are
\emph{not analytic in the physical strip} $\theta \in \mathbb{R} +i(0,\pi)$,
we have $[\phi(f), \phi'(g)] \neq 0$, namely, even the weak commutativity fails for $\phi,\phi'$.
The goal of the present paper is to find alternative wedge-observables for
the S-matrix of the sine-Gordon model.

\section{The bound state operator}\label{chi}
We introduce an operator $\chi(f)$ similarly to \cite{CT16-diag}, which we again call the ``bound state operator'',
whose mathematical structure corresponds to our fusion table, which is same as
the breather-breather fusion processes in the sine-Gordon model with two breathers.
In this model, the ``elementary particle'' is $b_1$, and we restrict ourselves
to the case where $f_{b_1}$ is the only non-zero component of a test function $f$.

\subsection{Definitions and domains}
We define it as an unbounded operator on the $S$-symmetric Fock space $\H$.
Recall that for $s < t$, $H^2(\SS_{s,t})$ is the Hardy space of analytic functions $\Psi$ in
$\SS_{s,t} := \RR + i(s,t)$ such that $\Psi(\theta + i\l)$ is $L^2(\RR)$ as a function of $\theta$
for each $\l \in (s,t)$ and their $L^2$-norm is uniformly bounded for $\l$.
For a multi-component test function $f$ whose only non-zero component is $f_{b_1}$
and is supported in $W_\L$,
its action on $\H_1$ is given as follows:
\begin{align}
 \dom(\chi_{1}(f)) &:=  H^2\left(\SS_{-\theta_{(b_1 b_1)}^{b_2},0} \right) \oplus H^2\left(\SS_{-\theta_{(b_2 b_1)}^{b_1},0} \right) \nonumber \\
 (\chi_{1}(f)\xi)_{b_k} (\theta) &:= \left\{\begin{array}{ll}
                                          -i \eta^{b_1}_{b_1 b_2} f^+_{b_1} (\theta + i\theta_{(b_1 b_2)}^{b_1} ) \xi_{b_2} (\theta - i\theta_{(b_2 b_1)}^{b_1}) & \text{ if } k = 1,\\ \\
-i \eta^{b_2}_{b_1 b_1} f^+_{b_1} (\theta + i\theta_{(b_1 b_1)}^{b_2} ) \xi_{b_1} (\theta - i\theta_{(b_1b_1)}^{b_2}) & \text{ if } k = 2.
                                         \end{array}\right. \label{eq:chi1}
\end{align}
where $\eta^{b_2}_{b_1 b_1}$, $\eta^{b_1}_{b_1 b_2}$ are the matrix elements introduced in Sec.~\ref{scattering}, see  Eq.~\eqref{Reseta}.
Actually, $\theta_{(b_2 b_1)}^{b_1} = \theta_{(b_2 b_1)}^{b_2} = \frac{\pi\nu}2$,
hence $\dom(\chi_{1}(f)) =  H^2\left(\SS_{-\frac{\pi\nu}2,0} \right)^{\oplus 2}$,
but we often keep the notation above for homogeneity.

%
The full operator $\chi(f)$ is the direct sum of its components $\chi_n(f)$ on $\H_n$:
\begin{align} \label{chin}
 \chi_n(f) := nP_n(\chi_1(f)\otimes\1\otimes\cdots\otimes\1)P_n, \quad
 \chi(f) = \bigoplus_{n=0}^\infty \chi_n(f).
\end{align}
Similarly, and as in \cite{CT16-diag}, we introduce the reflected bound state operator $\chi'(g)$
for a test function $g$ supported in the right wedge $W_\R$.
Again, its one particle projection for $g$ having only one non-zero component $g_{b_1}$ is given by
\begin{align*}
 \dom(\chi'_{1}(g)) &:= H^2\left(\SS_{0,\theta_{(b_1 b_1)}^{b_2}} \right) \oplus H^2\left(\SS_{0,\theta_{(b_2 b_1)}^{b_1}} \right) \nonumber \\
 (\chi'_{1}(g)\xi)_{b_k} (\theta) &:= \left\{\begin{array}{ll}
                                          -i \eta^{b_1}_{b_1 b_2} g^+_{b_1} (\theta - i\theta_{(b_1 b_2)}^{b_1} ) \xi_{b_2} (\theta + i\theta_{(b_2 b_1)}^{b_1}) & \text{ if } k = 1, \\ \\
 -i \eta^{b_2}_{b_1 b_1} g^+_{b_1} (\theta - i\theta_{(b_1 b_1)}^{b_2} ) \xi_{b_1} (\theta + i\theta_{(b_1 b_1)}^{b_2}) & \text{ if } k = 2.
                                         \end{array}\right.
\end{align*}
The full operator on $\H$ is given by
\begin{align}\label{chipn}
\chi'_n(g) := nP_n(\1\otimes\cdots\otimes\1\otimes\chi'_1(g))P_n, \quad
\chi'(g) = \bigoplus_n \chi'_n(g).
\end{align}

This operator is related to $\chi$ by the CPT operator $J$:
\begin{equation*}
 \chi'(g) = J\chi(g_j)J.
\end{equation*}
To see this, let us consider the one-particle components. By recalling
the expression \eqref{eq:chi1},
\begin{align*}
(J \chi_1(g_j)J \xi)_{b_\ell} (\theta)
 &= \overline{(\chi_1(g_j)J \xi)_{ b_\ell }(\theta)} \nonumber\\
 &=  \overline{-i\eta_{b_1 b_k}^{b_\ell } (g_j)^+_{b_1} (\theta + i\theta^{b_\ell}_{(b_1 b_k)}) (J\xi)_{b_k} (\theta - i\theta^{b_\ell}_{(b_k b_1)})}\nonumber\\
 &=  -i\eta^{b_\ell}_{b_1 b_k} g^+_{b_1}(\theta - i\theta^{b_\ell}_{(b_1 b_k)}) \xi_{b_k}(\theta +i\theta^{b_\ell}_{(b_k b_1)}) \nonumber\\
&= (\chi'(g)\xi )_{b_\ell} (\theta),
\end{align*}
where $l = 1$ or  $2$ and $k = 2$ or $1$, respectively, and we used that $-i\eta_{\a\b}^\g \in \RR$.
As $J_ n$ commutes with $P_n$, we have $\chi'_n(g) = J_n \chi_n(g_j) J_n$. Since the whole operators $\chi(g)$ and $\chi'(g)$ are defined as the direct sum, the desired equality follows.

We give some more explicit expressions of Eq.ns~\eqref{chin} and \eqref{chipn} by applying them to a $n$-particle vector which we assume to be $S$-symmetric and in the domain of $\chi_1(f)\otimes\1\otimes\cdots\otimes\1$ and of  $\1\otimes\cdots\otimes\1\otimes\chi'_1(g)$, respectively.
We have, from \ref{bootstrap}, \ref{parity} and \ref{hermitian} exactly as in \cite[Section 3.2]{CT16-diag},
\begin{align}\label{chiextended}
& ( \chi(f)\Psi_n )^{b_{k_1} \cdots b_{k_n}}(\theta_1, \ldots, \theta_n) \nonumber\\
&\quad =-i\sum_{1\leq \ell \leq n, \;\a_\ell =b_1,b_2}  \eta^{b_{k_\ell}}_{b_1 \a_\ell}\left( \prod_{1\leq j\leq \ell -1 } S^{b_{k_j}  b_1}_{b_1 b_{k_j}}(\theta_\ell - \theta_j +i\theta^{b_{k_\ell}}_{(b_1 \a_\ell)})\right) \nonumber\\
&\qquad\qquad \times f^+_{b_1} (\theta_\ell + i\theta^{b_{k_\ell}}_{(b_1 \a_\ell)})\Psi_n^{b_{k_1} \ldots b_{k_{\ell -1}} \a_\ell  b_{k_{\ell +1}} \ldots b_{k_n}}\left(\theta_1,\cdots,\theta_{\ell -1},\theta_\ell -i\theta_{(\a_\ell b_1)}^{b_{k_\ell}},\theta_{\ell +1},\cdots\theta_n\right),
\end{align}
where $k_1, \ldots, k_n = 1,2$ and we applied our convention that $\eta_{\a\b}^\g = 0$
if $(\a\b) \to \g$ is not a fusion, and terms containing such $\eta_{\a\b}^\g$
should be ignored (even if it contains expressions such as $\Psi(\cdots,\theta - i\theta_{(\b\a)}^\g,\cdots)$
which can be meaningless, as it might be outside the domain of analyticity).

We have a similar expression for $\chi'(g)$:
\begin{align}\label{chipextended}
&(\chi'(g)\Psi_n)^{b_{k_1} \ldots b_{k_n}}(\theta_1, \ldots, \theta_n)=\nonumber\\
&\quad = -i\sum_{1\leq \ell \leq n, \;\a_\ell =b_1,b_2}  \eta_{b_1 \a_\ell}^{b_{k_\ell}}\left( \prod_{\ell +1 \leq j\leq n } S^{ b_1 b_{k_j}}_{b_{k_j} b_1}(\theta_j - \theta_\ell +i\theta^{b_{k_\ell}}_{(b_1 \a_\ell)})\right)  \nonumber\\
&\qquad \qquad \times g^+_{b_1} (\theta_\ell - i\theta^{b_{k_\ell}}_{(b_1 \a_\ell)}) \Psi_n^{b_{k_1} \ldots b_{k_{\ell -1}} \alpha_{\ell} b_{k_{\ell +1}} \ldots b_{k_n}}\left(\theta_1,\cdots,\theta_{\ell -1},\theta_{\ell} +i\theta_{(\a_\ell b_1)}^{b_{k_\ell}},\theta_{\ell +1} \cdots\theta_n\right).
\end{align}

\subsection{Some properties}\label{sec:prop}

We remark here on some of the properties of $\chi(f)$, noting that
analogous properties hold by construction for $\chi'(g)$.
For a multi-component real test function $f$ whose only non-zero component is $f_{b_1}$ which is real,
we can prove that $\chi(f)$ is densely defined and symmetric as follows.

By construction, $\chi_1(f)$ is clearly densely defined.
To show that $\chi_1(f)$ is symmetric, we take two vectors $\xi, \psi \in \operatorname{Dom}(\chi_1 (f))$
whose components have compact inverse Fourier transform. One can show that these vectors form a core for $\chi_1(f)$.
By recalling that $\eta^{b_\ell}_{b_1 b_k} = 0$ unless $k=1, \ell=2$ or $k=2, \ell=1$,
we compute on vectors $\xi, \psi$ from the core:
\begin{align*}
\langle \psi, \chi_{1}(f) \xi \rangle
&= -\sum_{k,\ell} i\eta^{b_\ell}_{b_1 b_k} \int d\theta\, \overline{\psi^{b_\ell} (\theta)} f^+_{b_1} (\theta +i\theta^{b_\ell}_{(b_1 b_k)}) \xi^{b_k}(\theta -i\theta^{b_\ell}_{(b_k b_1)}) \nonumber \\
&= -\sum_{k,\ell} i\eta^{b_\ell}_{b_1 b_k} \int d\theta\, \overline{f^+_{ b_1} (\theta +i\pi -i\theta^{b_\ell}_{(b_1 b_k)})} \overline{\psi^{b_\ell} (\theta)}  \xi^{b_k}(\theta -i\theta^{b_\ell}_{(b_k b_1)}) \nonumber \\
&= -\sum_{k \ell} i\eta^{b_\ell}_{b_1 b_k} \int d\theta\, \overline{f^+_{ b_1} (\theta +i\pi -i\theta^{b_\ell}_{b_1 b_k})} \overline{\psi^{b_\ell} (\theta -i\theta^{b_\ell}_{(b_k b_1)})}  \xi^{b_k}(\theta) \nonumber\\
&= -\sum_{k \ell} i\eta^{b_k}_{b_1 b_\ell} \int d\theta\, \overline{f^+_{b_1} (\theta +i\theta^{b_k}_{(b_1 b_\ell)})} \overline{\psi^{b_\ell} (\theta -i\theta^{b_k}_{(b_\ell b_1)})}  \xi^{b_k}(\theta) = \langle \chi_{1}(f)\psi, \xi \rangle,
\end{align*}
where in the second equality we used the property $f^+(\theta + i\l) = \overline{f^+(i\pi - \theta - i\l)}$
explained at the end of Sec.~\ref{Fock}. In the third equality we used Cauchy theorem and performed the shift $\theta \rightarrow \theta + i\theta^{b_{\ell}}_{(b_k b_1)}$, since the integrand is analytic, bounded and rapidly decreasing in the strip $\mathbb{R} +i(0,\pi)$ due to $\xi, \psi$ being the Fourier transforms of compactly supported functions and the properties of $f^+$.
In the fourth equality we used the properties
$\pi - \theta^{b_2}_{b_1 b_1} = \theta^{b_1}_{( b_1 b_2)}$, $\theta^{b_2}_{(b_1 b_1)} = \theta^{b_1}_{(b_2 b_1)}$ and $\eta^{b_2}_{b_1 b_1} =\eta^{ b_1}_{b_1 b_2}$ from Sec.~\ref{scattering}. 

We can show that $\chi_n(f)$ is densely defined and symmetric by arguing as in \cite[Proposition 3.1]{CT16-diag}.

Furthermore, the operator $\chi(f)$ is covariant with respect to the action $U$ of the Poincar\'e group $\poincare$ on $\H$ that we introduced in Section~\ref{Fock} in the following sense. For a test function $f$ supported in $W_\L$ and $(a,\lambda) \in \poincare$ such that $a \in W_\L$, we can show that $\operatorname{Ad}U(a, \lambda)(\chi(f)) \subset \chi(f_{(a, \lambda)})$.
The key to the proof are the relations \eqref{eq:massfusion}, see \cite[Proposition 3.2]{CT16-diag} for details.

\section{Weak commutativity}\label{sec:comm}

We introduce the field
\begin{equation*}
\tilde \phi(f) = \phi(f) + \chi(f)
\end{equation*}
and its reflected field $\tilde \phi'(g) = \phi'(g) + \chi'(g) = J \tilde \phi (g_j) J$ in a similar manner as in \cite{CT16-diag}. For $f$  with support in $W_\L$ and such that $f^*= f$, the field $\tilde \phi(f)$ fulfills the properties listed in \cite[Proposition 4.1]{CT16-diag}, and a similar result also holds for the reflected field $\tilde \phi'(g)$.  Regarding the domain of $\tilde \phi$, we note that, since the domain of $\chi(f)$ contains vectors with finite particle number and with certain analyticity and boundedness properties (see Sec.~\ref{chi}), its domain is included in the domain of $\phi(f)$, and therefore $\operatorname{Dom}(\tilde \phi(f)) = \operatorname{Dom}(\chi(f))$.

As already mentioned in \cite{CT16-diag}, the field $\tilde \phi(f)$ has very subtle domain properties. In particular, because of the poles of $S$, after applying this operator to a vector (not the vacuum) in its domain, it generates a vector which is no longer in the domain of $\tilde \phi'(g)$. For this reason, products of the form $\tilde \phi(f) \tilde \phi'(g)$ and $\tilde \phi'(g)\tilde \phi(f)$ are not well defined, and we need to compute the commutator $[\tilde \phi(f), \tilde \phi'(g)]$ between arbitrary vectors $\Phi, \Psi$ from a suitable space (see below). Moreover, the commutator is smeared with test functions $f,g$ with only nonzero components corresponding to $b_1$.

We start by considering vectors $\Psi_n^{\pmb{b_k}}$
in the domain discussed in Sec.~\ref{sec:prop}. These vectors admit analytic continuation in the first variable,
and actually a meromorphic continuation in each variable, to $\pm i\frac{\pi \nu}2$.
We also note that for certain components $\Psi_n^{b_{k_1} \cdots b_{k_n}}(\theta_1, \cdots, \theta_n)$,
specifically in the case where two of the indices are equal, $b_{k_j} = b_{k_\ell} = \a$, we can infer the existence of zeros by the following
computation:
\begin{align*}
 &\Psi_n^{b_{k_1} \cdots \a \cdots \a \cdots b_{k_n}}(\theta_1,\cdots,\theta_j,\cdots,\theta_\ell,\cdots,\theta_n)\\
 &= \left(  \prod_{p = j+1}^{\ell-1} S^{\a b_{k_p}}_{b_{k_p} \a } (\theta_p-\theta_j)S^{b_{k_p} \a}_{\a b_{k_p}}(\theta_\ell - \theta_p) \right)   S_{\a\a}^{\a\a}(\theta_\ell-\theta_j) \\
&\quad \times  \Psi_n^{b_{k_1} \cdots \a \cdots \a \cdots b_{k_n}}(\theta_1,\cdots,\theta_\ell,\cdots,\theta_j,\cdots,\theta_n).
\end{align*}
Hence, by \ref{atzero} and  \ref{hermitian}, $\Psi_n^{\pmb{b_k}}$ has a zero at $\theta_j - \theta_\ell = 0$.
However, this does not imply existence of zeros for other components.
Furthermore, in the proof of Theorem \ref{theo:commutator},
we will encounter certain poles of $S$ in the computation.
Hence, we consider vectors from the following space:
\begin{equation}\label{domain}
\D_0 := \left\{ \Psi \in \D: \begin{array}{l}
                              \Psi_n^{\pmb{b_k}} \text{ is analytic in } \mathbb{R}^n +i(-\frac{\pi \nu}2,\frac{\pi \nu}2)^n,\\
                              \Psi_n^{\pmb{b_k}}(\pmb\theta + i\pmb \l) \in L^2(\RR^n) \text{ for } \pmb \l \in (-\frac{\pi \nu}2, \frac{\pi \nu}2)^n, \text{ with a uniform bound and}\\
                              \text{has a zero at } \theta_j -\theta_\ell = 0, \pm i\pi(1-\nu), \pm i\pi(\frac{3\nu}2 - 1), \pm \frac{i\pi \nu}{2} \text{ for all } j, \ell
                             \end{array}
        \right\},
\end{equation}
where $k_j=1,2$. Note that $\D_0 \subset \dom(\fct(f)) \cap \dom(\fct'(g))$.

One can see that $\D_0$ is dense as follows:
we take 
\[
C_n(\pmb{\theta})
:= \prod_{\l \in \Lambda}\prod_{1 \leq j < k \leq n} \frac{(\theta_k -\theta_j -i\l)(\theta_j -\theta_k -i\l)}{(\theta_k -\theta_j -2\pi i)(\theta_j -\theta_k -2\pi i)},
\quad \Lambda = \left\{0, \pi(1-\nu), \pi(\textstyle{\frac{3\nu}2-1)}, \frac{\pi\nu}2 \right\},
\]
and consider the set
\[
 \left\{M_{C_n}P_n(\xi_1\otimes\cdots\otimes\xi_n), \xi_j \in \dom(\chi_1(f))\cap\dom(\chi'_1(g))\right\}. 
\]
As $C_n$ is symmetric and it has zeros at the poles of $S$, the set above is a subset of $\D_0$.
Furthermore, as $C_n$ is bounded and invertible on $\RR^n$, $M_{C_n}$ maps a dense set to a dense set.
The set $\left\{P_n(\xi_1\otimes\cdots\otimes\xi_n), \xi_j \in \dom(\chi_1(f))\cap\dom(\chi'_1(g))\right\}$ is dense,
therefore, so are its image $M_{C_n}\left(\dom(\fct(f)) \cap \dom(\fct'(g))\right)$ and $\D_0$.
Thanks to \ref{regularity}, \cite[Proposition E.7]{Tanimoto16-1} and the properties of $\D_0$,
we can safely use analytic continuations in the proof of our main theorem.

\begin{theorem}\label{theo:commutator}
Let $f$ and $g$ be test functions supported in $W_\L$ and $W_\R$, respectively, and with the property that $f=f^*$ and $g=g^*$.
Furthermore, assume that $f,g$ have components $f_{b_k} =0$ and $g_{b_k} =0$ for $k \neq 1$.
Then, for each $\Phi, \Psi$ in $\D_0$, we have
 \[
 \<\fct(f)\Phi, \fct'(g)\Psi\> = \<\fct'(g)\Phi, \fct(f)\Psi\>.
 \]
\end{theorem}
\begin{proof}
As in our previous works, we may assume that the vectors $\Phi$ and $\Psi$ are already $S$-symmetric.
Furthermore, we recall that the domains of $\fct(f), \fct'(g)$ coincide with those of $\chi(f), \chi'(g)$,
respectively, hence we have the following equalities as operators:
 \begin{align*}
  \fct(f) &= \phi(f) + \chi(f) = z^\dagger(f^+) + \chi(f) + z(J_1f^-), \\
  \fct'(g) &= \phi'(g) + \chi'(g) = z'^\dagger(g^+) + \chi'(g) + z'(J_1g^-).
 \end{align*}
Therefore, the (weak) commutator $[\fct(f), \fct'(g)]$ expands into several terms that we will compute individually.

\begin{flushleft} 
{\it The commutator $[\phi(f), \phi'(g)]$}
\end{flushleft}

This commutator has been computed in \cite{LS14} and then simplified in the case where $S$ is diagonal in \cite{CT16-diag}. Here, we briefly recall its expression:
\begin{align*}
&([\phi'(g), \phi(f)]\Psi_n)^{\pmb{b_k}}(\theta_1,\cdots,\theta_n) \nonumber\\
&\quad=\;\int d\theta'\, \left(  g^-_{b_1}(\theta') \left(\prod_{p=1}^n S^{b_{k_p} b_1}_{b_1 b_{k_p}}(\theta' -\theta_p)\right) f^+_{b_1}(\theta')
- g^+_{b_1}(\theta') \left(\prod_{p=1}^n  \overline{S^{b_{k_p} b_1}_{b_1 b_{k_p}}(\theta' -\theta_p)}\right) f^-_{b_1}(\theta')\right) \nonumber\\
&\quad\quad\quad\quad\times(\Psi_n)^{\pmb{b_k}}(\theta_1, \ldots, \theta_n).
\end{align*}
By \ref{crossing} and the analytic properties of $f^\pm, g^\pm$ explained in Section \ref{Fock},
the first term in the integrand is equal to the second term up to a shift of $+i\pi$ in $\theta'$.
Since $S$ has some poles in the physical strip, we obtain residues from this difference.

We are considering test functions $f,g$ whose only non-zero components correspond
to $b_1$. In this case, the factor $S_{b_1 b_k}^{b_k b_1}$ appearing in the expression of
the commutator have exactly two simple poles at $\zeta = i\theta_{b_1 b_k}^{b_{k'}}, i\theta_{b_1 b_k}^{\prime b_{k'}}$ with $k=1, k'=2$
and $k = 2, k' = 1$,
as seen in the fusion table in Sec.~\ref{scattering}.

With the notation $R_{b_1 b_k}^{b_{k'}}, R_{b_1 b_k}^{\prime b_{k'}}$ which are nonzero only for $k=1, k' =2$ and $k=2, k' = 1$, by applying the Cauchy theorem, we get the contributions from the above-mentioned poles:
\begin{align*}
&\frac{1}{2\pi i} ([\phi'(g),\phi(f)]\Psi_n)^{\pmb{b_k}}(\theta_1, \ldots, \theta_n)\nonumber\\
&\; =\; \sum_{k =1,2} \left( \sum_{j=1}^n R_{b_1 b_{k_j}}^{b_k} g^-_{ b_1}(\theta_j +i\theta_{b_1 b_{k_j}}^{b_k}) f^+_{b_1} (\theta_j +i\theta_{ b_1 b_{k_j}}^{b_k}) \left( \prod_{\substack{p=1\\ p \neq j}}^n S^{b_{k_p} b_1}_{b_1 b_{k_p}}(\theta_j +i\theta_{b_1 b_{k_j}}^{b_k} -\theta_p) \right)\right.\nonumber\\
&\; \quad\quad\quad\quad+\; \left.\sum_{j=1}^n R_{ b_1 b_{k_j}}^{\prime b_k} g^-_{ b_1}(\theta_j +i\theta'\phantom{}^{b_k}_{b_1 b_{k_j}}) f^+_{b_1} (\theta_j +i\theta'\phantom{}^{b_k}_{ b_1 b_{k_j}}) \left( \prod_{\substack{p=1\\ p \neq j}}^n S^{b_{k_p} b_1}_{b_1 b_{k_p} }(\theta_j +i\theta'\phantom{}^{b_k}_{b_1 b_{k_j}}-\theta_p) \right)\right) \nonumber \\
&\quad\quad\quad\times (\Psi_n)^{b_{k_1} \ldots b_{k_n}}(\theta_1, \ldots,\theta_n).
\end{align*}
More explicitly, the possible terms from the above expression are given by the following.

\allowdisplaybreaks
\begin{subequations}
\begin{align}
&\frac{1}{2\pi i} ([\phi'(g),\phi(f)]\Psi_n)^{\pmb{b_k}}(\theta_1, \ldots, \theta_n)\nonumber\\
&\;=\; \sum_{j=1}^n R^{b_2}_{b_1 b_1} g^-_{b_1} (\theta_j +i\theta^{b_2}_{b_1 b_1})f^+_{b_1}(\theta_j +i\theta^{b_2}_{b_1 b_1}) \left( \prod_{\substack{p=1\\ p \neq j}}^n S^{b_{k_p} b_1}_{b_1 b_{k_p}}(\theta_j +i\theta_{b_1 b_1}^{b_2}-\theta_p)  \right)\nonumber\\
&\quad \quad \times (\Psi_n)^{b_{k_1} \ldots b_1 \ldots b_{k_n}}(\theta_1, \ldots,\theta_j, \ldots, \theta_n) \label{phi1}\\
&\; +\; \sum_{j=1}^n R'\phantom{}^{b_1}_{b_1 b_2} g^-_{b_1} (\theta_j +i\theta'\phantom{}^{b_1}_{b_1 b_2})
f^+_{b_1}(\theta_j +i\theta'\phantom{}^{b_1}_{b_1 b_2}) 
\left( \prod_{\substack{p=1\\ p \neq j}}^n S^{b_{k_p} b_1}_{b_1 b_{k_p}}(\theta_j +i\theta'\phantom{}^{b_1}_{b_1 b_2}-\theta_p)  \right)\nonumber\\
&\quad \quad \times (\Psi_n)^{b_{k_1} \ldots  b_2 \ldots b_{k_n}}(\theta_1, \ldots, \theta_j, \ldots, \theta_n) \label{phi2} \\
&\;+\; \sum_{j=1}^n R^{b_1}_{b_1 b_2} g^-_{b_1} (\theta_j +i\theta^{b_1}_{b_1 b_2})f^+_{b_1}(\theta_j +i\theta^{b_1}_{b_1 b_2}) \left( \prod_{\substack{p=1\\ p \neq j}}^n S^{b_{k_p} b_1}_{b_1 b_{k_p}}(\theta_j +i\theta_{b_1 b_2}^{b_1}-\theta_p)  \right)\nonumber\\
&\quad \quad \times (\Psi_n)^{b_{k_1} \ldots b_2 \ldots b_{k_n}}(\theta_1, \ldots, \theta_j, \ldots, \theta_n)\label{phi4} \\
&\; +\; \sum_{j=1}^n R'\phantom{}^{b_2}_{b_1 b_1} g^-_{b_1} (\theta_j +i\theta'\phantom{}^{b_2}_{b_1 b_1})f^+_{b_1}(\theta_j +i\theta'\phantom{}^{b_2}_{b_1 b_1}) \left( \prod_{\substack{p=1\\ p \neq j}}^n S^{b_{k_p} b_1}_{b_1 b_{k_p}}(\theta_j +i\theta'\phantom{}^{b_2}_{b_1 b_1}-\theta_p)  \right)\nonumber\\
&\quad \quad \times (\Psi_n)^{b_{k_1} \ldots b_1 \ldots b_{k_n}}(\theta_1, \ldots, \theta_j, \ldots, \theta_n)\label{phi3}.
\end{align}
\end{subequations}
%

\begin{flushleft} 
{\it The commutator $[\chi(f), \chi'(g)]$}
\end{flushleft}

We compute this commutator between vectors $\Psi, \Phi$ with only $n$-particle components and with $f,g$ having only non-zero components of type $b_1$.
Recall the expressions of $\chi(f)$ and $\chi'(g)$ in Sec.~\ref{chi},
where they are written as the sum of $n$ operators acting on different variables,
therefore, there are $n^2$ terms in each of
the scalar products $\<\chi'(g)\Phi, \chi(f)\Psi\>$ and $\<\chi(f)\Phi, \chi'(g)\Psi\>$.
Of these, one can show that the $n(n-1)$ terms in which the above-mentioned operators act on different variables
give exactly the same contribution, exactly as in \cite{CT16-diag} (this time the operators $\chi_1(f)$ and $\chi'_1(g)$
are not positive, but $\chi(f)\otimes\1\otimes\cdots \otimes\1$ and $\1\otimes\cdots\otimes\1\otimes\chi'_1(g)$ are strongly commuting,
hence we may consider their polar decomposition), which we denote by $C$, therefore, they cancel in the commutator and hence are irrelevant.

Following \cite[P.\! 35]{CT16-diag}, we exhibit the relevant parts ($\pmb{k} := k_1, \ldots, k_n$ where each $k_j$ can take $1,2$.
Furthermore, if $k_j = 1$, then we put $k'_j = 2$ and if $k_j = 2$, then $k'_j = 1$):
\begin{align*}
&\<\chi'(g)\Phi, \chi(f)\Psi\> - C \nonumber\\
& \;= \;  \sum_{j=1}^n \underset{\a_j, \b_j =1,2}{\sum_{\pmb{k}}}  \int d\theta_1 \ldots d\theta_n\; \eta^{b_{k_j}}_{b_1 b_{\a_j}} \left(  \prod_{p=1}^{j-1}S^{b_{k_p} b_1}_{b_1 b_{k_p}}  \left(\theta_j -\theta_p +i\theta_{(b_1 b_{\a_j})}^{b_{k_j}}\right)\right)f^+_{b_1}\left(\theta_j +i\theta_{(b_1 b_{\a_j})}^{b_{k_j}}\right) \nonumber\\
& \;\quad \times  (\Psi_n)^{b_{k_1} \ldots b_{\a_j} \ldots b_{k_n}}\left(\theta_1, \ldots, \theta_j - i\theta^{b_{k_j}}_{(b_{\a_j} b_1)}, \ldots, \theta_n\right) \eta_{b_1 b_{\b_j}}^{b_{k_j}} \left( \prod_{q= j+1}^n S^{b_{k_q} b_1}_{b_1 b_{k_q}}\left(\theta_j -\theta_q +i\theta_{(b_1 b_{\b_j})}^{b_{k_j}}   \right)  \right) \nonumber\\
& \;\quad \times g^+_{b_1}\left(\theta_j +i\theta_{(b_1 b_{\b_j})}^{b_{k_j}} -i\pi \right) \overline{(\Phi_n)^{b_{k_1} \ldots b_{\b_j} \ldots b_{k_n}}\left(\theta_1, \ldots, \theta_j +i\theta_{(b_{\b_j} b_1)}^{b_{k_j}}, \ldots, \theta_n \right)} \nonumber\\
& \;= \;  \sum_{j=1}^n \sum_{\pmb{k}}  \int d\theta_1 \ldots d\theta_n\; \eta^{b_{k_j}}_{b_1 b_{k'_j}} \left(  \prod_{p=1}^{j-1}S^{b_{k_p} b_1}_{b_1 b_{k_p}}  \left(\theta_j -\theta_p +i\theta_{(b_1 b_{k'_j})}^{b_{k_j}}\right)\right)f^+_{b_1}\left(\theta_j +i\theta_{(b_1 b_{k'_j})}^{b_{k_j}}\right) \nonumber\\
& \;\quad \times  (\Psi_n)^{b_{k_1} \ldots b_{k'_j} \ldots b_{k_n}}\left(\theta_1, \ldots, \theta_j - i\theta^{b_{k_j}}_{(b_{k'_j} b_1)}, \ldots, \theta_n\right) \eta_{b_1 b_{k'_j}}^{b_{k_j}} \left( \prod_{q= j+1}^n S^{b_{k_q} b_1}_{b_1 b_{k_q}}\left(\theta_j -\theta_q +i\theta_{(b_1 b_{k'_j})}^{b_{k_j}}   \right)  \right) \nonumber\\
& \;\quad \times g^+_{b_1}\left(\theta_j +i\theta_{(b_1 b_{k'_j})}^{b_{k_j}} -i\pi \right) \overline{(\Phi_n)^{b_{k_1} \ldots b_{k'_j} \ldots b_{k_n}}\left(\theta_1, \ldots, \theta_j +i\theta_{(b_{k'_j} b_1)}^{b_{k_j}}, \ldots, \theta_n \right)} \nonumber\\
& \; = \; \sum_{j=1}^n \sum_{\pmb{k}}  \int d\theta_1 \ldots d\theta_n\; \eta^{b_{k_j}}_{b_1 b_{k'_j}}\left(  \prod_{p=1}^{j-1}S^{b_{k_p} b_1}_{b_1 b_{k_p}}  \left(\theta_j -\theta_p +i\theta_{b_1 b_{k'_j}}^{b_{k_j}} \right)\right) f^+_{b_1}\left(\theta_j +i\theta_{b_1 b_{k'_j}}^{b_{k_j}}\right) \nonumber\\
& \;\quad \times  (\Psi_n)^{b_{k_1} \ldots b_{k'_j} \ldots b_{k_n}}(\theta_1, \ldots, \theta_j, \ldots, \theta_n) \eta_{b_1 b_{k'_j}}^{b_{k_j}} \left( \prod_{q= j+1}^n S^{b_{k_q} b_1}_{b_1 b_{k_q}}\left(\theta_j -\theta_q +i\theta_{b_{k'_j} b_1}^{b_{k_j}}   \right)  \right) \nonumber\\
& \;\quad \times g^+_{b_1}\left(\theta_j +i\theta_{b_1 b_{k'_j}}^{b_{k_j}} -i\pi \right) \overline{(\Phi_n)^{b_{k_1} \ldots b_{k'_j} \ldots b_{k_n}}\left(\theta_1, \ldots, \theta_j, \ldots, \theta_n \right)},
\end{align*}
where we used \eqref{chiextended} and \eqref{chipextended},
exploited that $\eta_{b_1b_1}^{b_2}, \eta_{b_1b_2}^{b_1}$ are the only nonzero combinations,
then performed the shift $\theta_j \rightarrow \theta_j +i\theta^{b_{k_j}}_{(b_{k'_j} b_1)}$ in the third equality
and used $\theta_{\a\b}^\g = \theta_{(\a\b)}^\g + \theta_{(\b\a)}^\g$.
This shift in $\theta_j$ is allowed by the analyticity and decay properties
of $f^+$, $g^+$ at infinity in the strip, \cite[Lemma B.2]{CT15-1} and
the property of $\Psi, \Phi \in \D_0$ explained before Theorem~\ref{theo:commutator}:
more precisely, depending on whether $b_{k_p} = b_1$ or $b_2$ (respectively for $b_{k_q}$),
$S_{b_1 b_{k_p}}^{b_{k_p} b_1}(\zeta)$ has a pole at $i\pi\nu$ and $i\pi(1-\nu)$, or at $i\frac{\pi\nu}2$ and $i(1-\frac{\pi\nu}2)$.
As $\theta_j \to \theta_j +i\theta^{b_{k_j}}_{(b_{k'_j} b_1)} = \theta_j + i\frac{\pi\nu}2$
(this does not depend on $b_{k_j}$: see Table \ref{Table}),
the integral contour might move across the pole when $\theta_j \to \theta_j + i\pi(1-\nu), \theta_j \to \theta_j + i\pi(\frac{3\nu}2 - 1)$
or $\theta_j \to \theta_j + i\frac{\pi\nu}2$,
depending on the combination of $b_{k_p}$ and $b_{k_j}$.
But these poles are cancelled by the zeros of $\Psi_n, \Phi_n \in \D_0$, hence the shift is legitimate
and the result is $L^1$ (the integral is the inner product of two $L^2$-functions).

Similarly, we can compute the other term $\<\chi(f) \Phi, \chi'(g)\Psi\>$ in the commutator $[\chi(f),\chi'(g)]$ and obtain:
\begin{align*}
&\<\chi(f)\Phi, \chi'(g)\Psi\> - C\nonumber\\
=& \sum_{j=1}^n  \underset{\a_j, \b_j = 1,2}{\sum_{\pmb{k}}}   \eta_{b_1 b_{\a_j}}^{b_{k_j}} \int d\theta_1 \ldots d\theta_n \overline{\left( \prod_{p=1}^{j-1}S^{b_{k_p} b_1}_{b_1 b_{k_p}}\left(\theta_j -\theta_p +i\theta_{(b_1 b_{\a_j})}^{b_{k_j}}\right) \right)}\nonumber\\
&\; \times \; \overline{f^+_{b_1}\left(\theta_j +i\theta_{(b_1 b_{\a_j})}^{b_{k_j}}\right)} \overline{ (\Phi_n)^{b_{k_1} \ldots b_{\a_j} \ldots b_{k_n}} \left(\theta_1, \ldots, \theta_k -i\theta_{(b_{\a_j} b_1)}^{b_{k_j}}, \ldots, \theta_n\right)} \nonumber\\
&\; \times\; \eta_{b_1 b_{\b_j}}^{b_{k_j}} \left( \prod_{q= j+1}^n S^{b_1 b_{k_q}}_{b_{k_q} b_1}\left(\theta_q -\theta_j +i\theta_{(b_1 b_{\b_j})}^{b_{k_j}}\right) \right)\nonumber \\
&\; \times \; g^+_{b_1}\left(\theta_j -i\theta_{(b_1 b_{\b_j})}^{b_{k_j}}\right)(\Psi_n)^{b_{k_1} \ldots b_{\b_j} \ldots b_{k_n}}\left(\theta_1, \ldots, \theta_j +i\theta_{(b_{\b_j} b_1)}^{b_{k_j}}, \ldots, \theta_n\right)\nonumber\\
=& \sum_{j=1}^n \sum_{\pmb{k}} \eta^{b_{k_j}}_{b_1 b_{k'_j}}\int d\theta_1 \ldots d\theta_n \left( \prod_{p=1}^{j-1}S^{b_{k_p} b_1}_{b_1 b_{k_p}}\left(\theta_j -\theta_p -i\theta_{b_{k'_j} b_1}^{b_{k_j}} +i\pi\right) \right)\nonumber\\
&\; \times \; f^+_{b_1}\left(\theta_j -i\theta_{b_{k'_j} b_1}^{b_{k_j}}+i\pi\right) \overline{ (\Phi_n)^{b_{k_1} \ldots b_{k'_j} \ldots b_{k_n}} \left(\theta_1, \ldots, \theta_j, \ldots, \theta_n\right)}\nonumber\\
&\; \times\; \eta_{b_1 b_{k'_j}}^{b_{k_j}} \left( \prod_{q= j+1}^n S^{b_{k_q} b_1}_{b_1 b_{k_q}}\left(\theta_j -\theta_q -i\theta_{b_1 b_{k'_j}}^{b_{k_j}}+i\pi\right) \right)g^+_{b_1}\left(\theta_j -i\theta_{b_1 b_{k'_j}}^{b_{k_j}}\right) \nonumber\\
&\; \times \; (\Psi_n)^{b_{k_1} \ldots b_{k'_j} \ldots b_{k_n}}(\theta_1, \ldots, \theta_j, \ldots, \theta_n),
\end{align*}
where we used \eqref{chiextended}, \eqref{chipextended} and $\theta_{\a\b}^\g = \theta_{(\a\b)}^\g + \theta_{(\b\a)}^\g$,
we performed the shift $\theta_j \rightarrow \theta_j -i\theta^{b_{k_j}}_{(b_{k'_j} b_1)}$
and we used properties \ref{unitarity}--\ref{crossing}.
As before, we can perform the shift in $\theta_j$ using the analyticity and decay properties of $f^+, g^-$ at infinity in the strip,
\cite[Lemma B.2]{CT15-1} and the zeros of the vectors $\Psi, \Phi \in \D_0$. This also guarantees the fact that the result is still $L^1$.

Since there are only two types of fusion processes $(b_1 b_1) \rightarrow b_{2}$ and $(b_1 b_2) \rightarrow b_1$ in the model,
the possible contributions to the expectation values above are
\allowdisplaybreaks
\begin{subequations}
\begin{align}
&\<\chi'(g)\Phi, \chi(f)\Psi\> - C \nonumber\\
& \; = \;   \sum_{j=1}^n\sum_{\pmb{k}} \eta^{b_2}_{b_1 b_1}  \eta_{b_1 b_1}^{b_2} \int d\theta_1 \ldots d\theta_n\; \prod_{p=1}^{j-1}S^{b_{k_p} b_1}_{b_1 b_{k_p}}  \left(\theta_j -\theta_p +i\theta_{b_1 b_1}^{b_2}\right) \nonumber\\
& \; \; \times f^+_{b_1}\left(\theta_j +i\theta_{b_1 b_1}^{b_2}\right) (\Psi_n)^{b_{k_1} \ldots b_1 \ldots b_{k_n}}\left(\theta_1, \ldots, \theta_j , \ldots, \theta_n\right) 
\prod_{q= j+1}^n S^{b_{k_q} b_1}_{b_1 b_{k_q}}\left(\theta_j -\theta_q +i\theta_{b_1 b_1}^{b_2}\right)  \nonumber\\
 & \; \; \times g^+_{b_1}\left(\theta_j +i\theta_{b_1 b_1}^{b_2} -i\pi \right) \overline{(\Phi_n)^{b_{k_1} \ldots b_1 \ldots b_{k_n}}\left(\theta_1, \ldots, \theta_j, \ldots, \theta_n \right)} \label{chip1}\\
& \; + \;   \sum_{j=1}^n \sum_{\pmb{k}} \eta^{b_1}_{b_1 b_2}  \eta_{b_1 b_2}^{b_1} \int d\theta_1 \ldots d\theta_n\; \prod_{p=1}^{j-1}S^{b_{k_p} b_1}_{b_1 b_{k_p}}  \left(\theta_j -\theta_p +i\theta_{b_1 b_2}^{b_1}\right) \nonumber\\
& \; \; \times f^+_{b_1}\left(\theta_j +i\theta_{b_1 b_2}^{b_1}\right) (\Psi_n)^{b_{k_1} \ldots b_2 \ldots b_{k_n}}(\theta_1, \ldots, \theta_j, \ldots, \theta_n) 
 \prod_{q= j+1}^n S^{b_{k_q} b_1}_{b_1 b_{k_q}}\left(\theta_j -\theta_q +i\theta_{b_2 b_1}^{b_1}  \right) \nonumber\\
 & \; \; \times g^+_{b_1}\left(\theta_j +i\theta_{b_1 b_2}^{b_1} -i\pi \right) \overline{(\Phi_n)^{b_{k_1} \ldots b_2 \ldots b_{k_n}}\left(\theta_1, \ldots, \theta_j, \ldots, \theta_n \right)},\label{chip4}
\end{align}
\end{subequations}
and similarly,
\allowdisplaybreaks
\begin{subequations}
\begin{align}
&\<\chi(f)\Phi, \chi'(g)\Psi\> - C\nonumber\\
& \; = \; \sum_{j=1}^n \sum_{\pmb{k}}  \eta^{b_2}_{b_1 b_1}  \eta_{b_1 b_1}^{b_2} \int d\theta_1 \ldots d\theta_n\; \prod_{p=1}^{j-1}S^{b_{k_p} b_1}_{b_1 b_{k_p}}  \left(\theta_j -\theta_p -i\theta_{b_1 b_1}^{b_2} + i\pi\right) \nonumber\\
& \, \; \times f^+_{b_1}\left(\theta_j -i\theta_{b_1 b_1}^{b_2} + i\pi\right) \overline{(\Phi_n)^{b_{k_1} \ldots b_1 \ldots b_{k_n}}\left(\theta_1, \ldots, \theta_j,\ldots, \theta_n\right) }\nonumber\\
 & \, \; \times
\prod_{q= j+1}^n S^{b_{k_q} b_1}_{b_1 b_{k_q}}\left(\theta_j -\theta_q -i\theta_{b_1  b_1}^{b_2} +i\pi \right)   g^+_{b_1}(\theta_j -i\theta_{b_1 b_1}^{b_2} ) (\Psi_n)^{b_{k_1} \ldots b_1  \ldots b_{k_n}}(\theta_1, \ldots, \theta_j, \ldots, \theta_n )\label{chi1}\\
& \; + \; \sum_{j=1}^n \sum_{\pmb{k}} \eta^{b_1}_{b_1 b_2}  \eta_{b_1 b_2}^{b_1} \int d\theta_1 \ldots d\theta_n\; \prod_{p=1}^{j-1}S^{b_{k_p} b_1}_{b_1 b_{k_p}}  \left(\theta_j -\theta_p -i\theta_{b_2 b_1}^{b_1} +i\pi\right) \nonumber\\
& \, \; \times f^+_{b_1}(\theta_j -i\theta_{b_2 b_1}^{b_1} +i\pi) \overline{(\Phi_n)^{b_{k_1} \ldots b_2 \ldots b_{k_n}}\left(\theta_1, \ldots, \theta_j,\ldots ,\theta_n\right) }\nonumber\\
 & \, \; \times 
 \prod_{q= j+1}^n S^{b_{k_q} b_1}_{b_1 b_{k_q}}\left(\theta_j -\theta_q -i\theta_{ b_1 b_2}^{b_1} +i\pi  \right) g^+_{b_1}(\theta_j -i\theta_{b_1 b_2}^{b_1}  ) (\Psi_n)^{b_{k_1} \ldots b_2 \ldots b_{k_n}}(\theta_1, \ldots, \theta_j, \ldots, \theta_n )\label{chi3}.
\end{align}
\end{subequations}
Now, the commutator $[\phi'(g),\phi(f)]$ cancels the commutator $[\chi(f), \chi'(g)]$:
more precisely, \eqref{phi1} cancels \eqref{chip1}, \eqref{phi2} cancels \eqref{chi3},
\eqref{phi4} cancels \eqref{chip4}, \eqref{phi3} cancels \eqref{chi1}.
This uses the following properties:
\begin{itemize}
 \item The properties of fusion angles and residues,
such as $\theta_{b_1 b_2}^{b_1} := \theta_{(b_1 b_2)}^{b_1} + \theta_{(b_2 b_1)}^{b_1}$, $\theta_{b_1 b_1}^{\prime b_2} = \pi - \theta^{b_2}_{b_1 b_1}$, $\theta_{b_1 b_2}^{\prime b_1}= \pi - \theta^{b_1}_{b_2 b_1}$, $R'^{b_2}_{b_1 b_1} = - R^{b_2}_{b_1 b_1}$ and $R'^{b_1}_{b_2 b_1} = -R^{b_1}_{b_1 b_2}$.
 \item Eq.~\eqref{Reseta} and $R_{b_1 b_1}^{b_2}, R^{b_1}_{b_1 b_2} \in i\RR_+$, hence $(\eta_{b_1 b_1}^{b_2})^2 = - 2\pi iR_{b_1b_1}^{b_2}$ and $(\eta^{b_1}_{b_1 b_2})^2 = -2\pi iR^{b_1}_{b_1 b_2}$.
 \item  $f^+_{b_1}(\theta + i\pi) = f^-_{b_1}(\theta), g^+_{b_1}(\theta - i\pi) = g^-_{b_1}(\theta)$.
\end{itemize}
Most of these properties are from Section \ref{scattering}.


\begin{flushleft} 
 {\it The commutators $[\chi(f), z'(J_1g^-)]$ and $[z(J_1f^-), \chi'(g)]$}
\end{flushleft}

Using the expressions of $\chi(f)$ and $\chi'(g)$ in \eqref{chiextended} and \eqref{chipextended}, we can also compute these commutators as in \cite{CT16-diag}. Noting that $\eta^{b_2}_{b_1 b_1}, \eta_{b_1 b_2}^{b_1}$ are the only possible non-zero combinations, we find
\begin{align*}
&([\chi(f),z'(J_1g^-)]\Psi_n)^{b_{k_1} \ldots b_{k_{n-1}}}(\theta_1,\cdots,\theta_{n-1})\nonumber\\
&\; =\; \sqrt{n}\, i\eta^{b_1}_{b_1 b_2} \int d\theta'\, g^-_{b_1}(\theta') f^+_{b_1} (\theta' +i\theta_{(b_1 b_2)}^{b_1})
(\Psi_n)^{b_2 b_{k_1} \ldots b_{k_{n-1}}}(\theta' -i\theta_{(b_2 b_1)}^{b_1},\theta_1  \ldots \theta_{n-1})\nonumber\\
&\qquad \times  \left( \prod_{j=1}^{n-1}S^{b_{k_j} {b_1}}_{{b_1} b_{k_j}}(\theta' -\theta_j )\right),
\end{align*}
which it can be rewritten by shifting $\theta' \rightarrow \theta' +i\theta_{(b_2 b_1)}^{b_1}$ as follows
\begin{align}\label{chizp}
&([\chi(f),z'(J_1g^-)]\Psi_n)^{b_{k_1} \ldots b_{k_{n-1}}}(\theta_1,\cdots,\theta_{n-1})\nonumber\\
&\; =\; \sqrt{n} \,  i\eta^{{b_1}}_{b_1 b_2} \int d\theta'\, g^-_{b_1}(\theta' +i\theta_{(b_2 b_1)}^{b_1}) f^+_{b_1} (\theta' +i\theta_{b_1 b_2}^{b_1})
(\Psi_n)^{b_2 b_{k_1} \ldots b_{k_{n-1}}}(\theta' ,\theta_1  \ldots \theta_{n-1})\nonumber\\
&\qquad \times \left( \prod_{j=1}^{n-1}S^{b_{k_j} {b_1}}_{b_1 b_{k_j}}(\theta' +i\theta_{(b_2 b_1)}^{b_1} -\theta_j )\right).
\end{align}
For the shift in $\theta'$, as it is based on an application of the Cauchy Theorem,
it uses the analyticity and decay properties of $f^+, g^-$ at infinity in the strip,
\cite[Lemma B.2]{CT15-1} and the fact that the poles of the $S$-factors in the product
above are cancelled by the zeros of the vector $\Psi_n  \in \D_0$. More precisely, for $b_{k_j} = b_1$, $S_{b_1b_1}^{b_1b_1}(\zeta)$ has a pole at $\zeta = i\pi - i\theta_{b_1b_1}^{b_2} = i\pi(1-\nu)$. Noting that
$\pi(1-\nu) < \theta_{(b_2b_1)}^{b_1} = \frac{\pi \nu}2$ for $\frac23 < \nu < \frac45$, the zero of the factor $C_n$ at $i\pi (1- \frac{3\nu}{2})$ becomes relevant here (see below \eqref{domain}), while the pole at $\zeta = i \theta_{b_1 b_1}^{b_2} = i\pi \nu$ is not reached by the shift by $\frac{i\pi \nu}{2}$ in $\theta'$.
The pole of $S_{b_1b_2}^{b_2b_1}(\zeta)$ at $\zeta = i\frac{\pi \nu}2 =  i\theta_{(b_2 b_1)}^{b_1}$
is cancelled by the zeros of $\Psi_n$ arising from $S$-symmetry
(see the observations above \eqref{domain}), as in this case $b_{k_j} = b_2$, while the pole at $\zeta = i\pi(1 - \frac{\nu}{2})$ is not reached by the shift by $\frac{i\pi \nu}{2}$ in $\theta'$.

This also guarantees the fact that the result is still $L^2$.
Similarly, we have
\begin{align}
&([z(J_1f^-), \chi'(g)]\Psi_n)^{b_{k_1} \ldots b_{k_{n-1}}}(\theta_1,\cdots,\theta_{n-1})\nonumber\\
&\; =\; -\sqrt{n} \, i\eta_{b_1 b_2}^{{b_1}} \int d\theta'\, f^-_{{b_1}}(\theta') g^+_{b_1} (\theta' -i\theta_{(b_1 b_2)}^{b_1})
(\Psi_n)^{b_2 b_{k_1} \ldots b_{k_{n-1}}}(\theta' +i\theta_{(b_2 b_1)}^{b_1},\theta_1  \ldots \theta_{n-1})\nonumber\\
&\qquad \times  \left( \prod_{j=1}^{n-1}S^{b_1 b_{k_j} }_{b_{k_j} b_1}(\theta_j -\theta' +i\theta_{(b_1 b_2)}^{b_1})\right), \nonumber
\end{align}
and by shifting $\theta' \rightarrow \theta' -i\theta_{(b_2 b_1)}^{b_1}$ we can rewrite this expression as
\begin{align}\label{zchip}
&([z(J_1f^-), \chi'(g)]\Psi_n)^{b_{k_1} \ldots b_{k_{n-1}}}(\theta_1,\cdots,\theta_{n-1})\nonumber\\
&\; =\; -\sqrt{n} \,  i\eta_{b_1 b_2}^{b_1} \int d\theta'\, f^-_{b_1}(\theta' -i\theta_{(b_2 b_1)}^{b_1}) g^+_{b_1} (\theta' -i\theta_{b_1 b_2}^{b_1})
(\Psi_n)^{b_2 b_{k_1} \ldots b_{k_{n-1}}}(\theta',\theta_1  \ldots \theta_{n-1})\nonumber\\
&\qquad \times  \left( \prod_{j=1}^{n-1}S^{b_1 b_{k_j} }_{b_{k_j} {b_1}}(\theta_j -\theta' +i\theta_{b_1 b_2}^{b_1})\right) \nonumber\\
&\; =\; -\sqrt{n}\, i\eta_{b_2 b_1}^{b_1} \int d\theta'\, f^+_{b_1}(\theta'+ i\pi -i\theta_{(b_2 b_1)}^{b_1}) g^-_{b_1} (\theta' +i\pi -i\theta_{b_1 b_2}^{b_1})
(\Psi_n)^{b_2 b_{k_1} \ldots b_{k_{n-1}}}(\theta',\theta_1  \ldots \theta_{n-1})\nonumber\\
&\qquad \times  \left( \prod_{j=1}^{n-1}S^{b_1 b_{k_j} }_{b_{k_j} b_1}(\theta' - \theta_j +i\pi-i\theta_{b_1 b_2}^{b_1})\right),
\end{align}
where we used the property of $f^-, g^+$ under $\pi$-translation and \ref{crossing}.
As before, the shift in $\theta'$ is allowed as the poles of the $S$-factors in the product above
are cancelled by the zeros of $\Psi_n \in \D_0$. 
More precisely, $S_{b_2b_1}^{b_1b_2}(\zeta)$ has a pole at $i\frac{\pi\nu} 2$ and this is crossed as $\theta'$ is shifted by $i\pi(1-\frac\nu 2)$, hence the zero of the factor $C_n$ at $i\pi(1 - \nu)$ becomes relevant, while the pole at $\zeta = i\pi(1- \frac{\nu}{2})$  is cancelled by the zeros of $\Psi_n$ arising from $S$-symmetry.
The pole of $S_{b_1b_1}^{b_1b_1}(\zeta)$ at $\zeta = i\pi(1 - \nu)$ is crossed when $\theta'$ is shifted by $i\pi(1- \frac{\nu}{2})$, hence we need the zero of the factor $C_n$ at $i\frac{\pi \nu}{2}$ to compensate it, while the pole at $\zeta = i\pi \nu$ is not reached by the shift.

%
%
The commutators \eqref{chizp} and \eqref{zchip} cancel each other due to the property
$\pi - \theta^{b_1}_{b_1 b_2} = \theta^{b_1}_{(b_2 b_1)}$ (see Eq.~\eqref{prop:angles}).

\begin{flushleft} 
{\it The commutators $[z^\dagger(f^+), \chi'(g)]$ and $[\chi(f), z'^\dagger(g^+)]$}
\end{flushleft}

These commutators  are the adjoints of the previous ones,
therefore, they cancel weakly by the above computations.
\end{proof}

This shows the weak-commutativity property of the fields $\tilde \phi(f)$ and $\tilde \phi'(g)$. While being already a major step towards the construction of the model in the algebraic setting, it would be important to obtain a proof of strong commutativity of these fields in order to construct the corresponding wedge-algebras and to prove the existence of strictly local observables through intersection of a shifted right and left wedge. The proof of strong commutativity is however a hard task because of the subtle domain properties of $\tilde \phi(f)$ as mentioned at the beginning of Sec.~\ref{sec:comm}. We are in fact able to show that $\tilde \phi(f)$ is a symmetric quadratic form on a suitable domain of vectors, but it is not self-adjoint. Therefore, for the proof of strong commutativity, we would need not only to prove existence of self-adjoint extensions of the two fields, but also to select the ones that strongly commute. Some results in this direction are recently available in \cite{Tanimoto15-1, Tanimoto16-1} in the case of scalar S-matrices with bound states (e.g.\! the Bullough-Dodd model), but these techniques are hard to extend to more general S-matrices.

\begin{remark}\label{rm:cdd}
 Our proof depends only on the axioms and properties summarized in Section \ref{scattering}
 and not on the specific expressions of the S-matrix.
 This implies that our construction and the proof of weak commutativity work as well
 if one considers S-matrix such as
 \[
  S\phantom{}^{b_\ell b_k}_{b_k b_\ell}(\zeta) = S_{\text{SG}}\phantom{}^{b_\ell b_k}_{b_k b_\ell}(\zeta)\prod_{j=1}^{N}S_{j,\text{CDD}}\phantom{}^{b_\ell b_k}_{b_k b_\ell}(\zeta),
 \]
 where $S_{j,\text{CDD}}\phantom{}^{b_\ell b_k}_{b_k b_\ell}(\zeta)$ is a factor as
 in \eqref{Sb1} with (possibly different) parameters $\nu_{j,\pm}$, and $N$ is an odd number
 (this is necessary to maintain \ref{positiveresidue}). Therefore, we have abundant candidates for
 integrable QFT with the fusion structure considered in this paper.
\end{remark}

\section{Concluding remarks}\label{concluding}
We have investigated the construction of integrable models with bound states in a series of
two papers \cite{CT15-1, CT16-diag}. In the second paper the construction methods introduced in \cite{CT15-1}
are extended to a class of models with several particle species and ``diagonal'' S-matrices with poles
in the physical strip, which includes the $Z(N)$-Ising model and the affine-Toda field theories as examples.
This construction is based on finding observables localized in unbounded wedge-shaped regions
to avoid infinite series that characterize strictly local operators.
These strictly local observables, with some regularity condition on $S$, should be recovered
by taking intersection of the algebras generated by observables in right and left wedges (c.f.\! \cite{Lechner08, AL16}).

Here we considered a model which arises as a deformation of the massless sine-Gordon model 
with a parameter $\nu$ which corresponds to a certain range of the coupling constant,
$\frac{2}{3}<\nu< \frac45$, with an additional CDD factor.
As for the proof of weak wedge-locality, we need only some properties of the S-matrix components,
and there are abundant examples, as we pointed out in Remark \ref{rm:cdd}.
As far as we know, that QFTs with such S-matrices have never appeared in the literature.
It is an interesting problem to find (or exclude) a Lagrangian description of them
(note that the CDD factors appearing here are necessary and our S-matrix cannot be
considered as a perturbation of the sine-Gordon model in the sense of, e.g., \! \cite{SZ16}).
In this respect, let us observe that we could find the sign-adjusting CDD factor
only for the interval $\frac23 < \nu < \frac45$, while $\nu = 1$ corresponds to the (doubled)
Ising model. As there is a gap $\frac45 \le \nu < 1$, this casts doubt that a naive
perturbation argument should work.

The resulting theory describes two breathers $b_1, b_2$ subject to elastic scattering and with the property
that they can also fuse to form a bound state
(the fusion processes are $(b_1 b_1) \rightarrow b_2$, $(b_1 b_2) \rightarrow b_1$ and $(b_2 b_1) \rightarrow b_1$).
This model falls again into the class of  ``diagonal'' S-matrices, and in this sense,
it can be regarded as an extension of the previous techniques investigated in \cite{CT16-diag}.
This fusion table is the same as the restriction of the table of the Thirring model \cite{Smirnov92, BFKZ99} to the breather-breather sector
(note that it is called ``the sine-Gordon model''
in the literature in the form factor programme, e.g.~\cite{BFKZ99}, assuming the equivalence between them).
Yet, the original breather-breather S-matrix of the Thirring model does not satisfy
the positivity of residues (see Section \ref{scattering}), hence cannot be considered
as a separate model.
In this sense, the present paper highlights the really necessary properties of the S-matrix for wedge-locality
and contains a new hint in the construction of interacting quantum field theories in the algebraic framework.

An interesting problem would be an extension of such a construction to integrable models with
``non-diagonal'' S-matrices, e.g.\! the Thirring model \cite{BFKZ99} or $\mathrm{SU}(N)$-invariant S-matrices \cite{BFK08}.
It would be interesting to show that weak wedge-commutativity holds at least for
some of these models. They are currently under investigation.
It should be noted that commutation relations of pointlike fields have not been proved
for these models\footnote{Michael Karowski, private communication.}. Our methods represent a complementary way of proving existence of local observables, which may work if the S-matrix components
concerning elementary particles (solitons in the case of the Thirring model) have only simple poles, yet here
several analytic questions (such as the domains of unbounded operators and the modular nuclearity) must be addressed.

\begin{table}
\centering
\begin{tabular}{|c|c|l|}
 \hline
 Range of $\nu$ & The residue of pole of $S_{11}^{11}$ & Comment\\
 \hline
 $4/5<\nu < 1$ & $-i\RR_+$ & No adjusting CDD factor found \\
 \hline
 $ 2/3<\nu < 4/5$ & $-i\RR_+$ & Adjusting CDD factors found \\
 \hline
 $ 1/2 <\nu < 2/3$ &
$-i\RR_+$ & There are three breathers if one requires the \\
& & maximal analyticity within breathers. \\
& & No adjusting CDD factor found \\
 \hline
 $0<\nu < 1/2$ &
$i\RR_+$ & There is a breather $b_K$ for which \\
  & & $\displaystyle{\res_{\zeta = i\theta_{b_1b_k}^{b_{K+1}}}S_{b_1b_K}^{b_Kb_1}(\zeta) \in -i\RR_+}$ \\
 \hline
\end{tabular}
\caption{Ranges of the coupling constant $\nu$ in the sine-Gordon model}\label{coupling}
\end{table}

As we mentioned in Section~\ref{scattering}, the S-matrix studied in the present paper is a deformation of
the S-matrix of the sine-Gordon model in the range of the coupling constant $\frac{2}{3} <\nu< \frac45$ by a CDD factor.
The reason for the CDD factor is the following: while the fusion table of the breather-breather S-matrix is closed under fusions,
these S-matrix components cannot be considered as a separate model because
the residues of some poles in the physical strip are on $-i\RR_+$ (see comment before Eq.~\eqref{Reseta}),
which is not compatible with our proof. We note that also in the proof of local commutativity theorem
in the form factor programme \cite{Quella99} this property is used, therefore, it must be adjusted in some way.
Varying the range of the coupling constant $\nu$, the situation is as pictured in Table 2.
In particular, as explained in Sec.~\ref{scattering}, for $\frac45<\nu < 1$
there are no values of $\nu_-$ and $\nu_+$ which fulfill the required conditions after Eq.~\eqref{Sb2},
and our simplest form for a CDD factor does not work. For $\frac{1}{2}<\nu< \frac{2}{3}$ there are three breathers in the model
(if we take the maximal analyticity literally),
and both ${S_{\text{SG}}}_{11}^{11}$ and ${S_{\text{SG}}}_{12}^{21}$ have $s$-channel poles with residues in $-i\RR_+$.
We could not find a suitable CDD factor adjusting all the residues.
Finally, in the range $0<\nu< \frac{1}{2}$ there is an increasing number of breathers by maximal analyticity,
and while $\res_{\zeta = i\theta_{b_1b_1}^{b_2}}{S_{\text{SG}}}_{11}^{11}(\zeta) \in i\RR_+$,
there are other S-matrix components whose residues are in $-i\RR_+$.
We could not find a suitable CDD factor for this range as well.

Finally, the domain of the operator $\chi(f)$ is considerably small,
one can not only show that even the one-particle components $\chi_1(f)$ is not self-adjoint, see \cite{Tanimoto15-1},
but the domains of $\chi_n(f)$ must be somehow enlarged compensating the factor $C_n$.
We believe that these domain issues are fundamentally related with the complicated fusion processes of the models,
hence deserve a separate study.

\subsubsection*{Acknowledgements}
We thank Michael Karowski for informing us of the current status of the form factor programme.
Y.T.\! thanks Sabina Alazzawi and Wojciech Dybalski for the discussion on the relations
between Thirring and sine-Gordon models.

Y.T.\! is supported by the JSPS overseas fellowship.

{\small
\def\cprime{$'$} \def\polhk#1{\setbox0=\hbox{#1}{\ooalign{\hidewidth
  \lower1.5ex\hbox{`}\hidewidth\crcr\unhbox0}}}

}

\end{document}